\documentclass[UKenglish,runningheads]{scrartcl}

\usepackage{amsfonts}
\usepackage{mathtools}

\usepackage{url}
\usepackage{hyperref}

\usepackage[dvipsnames]{xcolor}
\usepackage{xcolor-solarized}

\usepackage{subcaption}

\usepackage{float}

\usepackage{threeparttable,booktabs}
\usepackage{tabularx}
\usepackage{makecell}

\usepackage[group-minimum-digits=4,group-separator={,}]{siunitx}
\DeclareSIUnit\persecond{/s}

\newcommand{\Oh}[1]{\mathcal{O}\left({#1}\right)}
\newcommand{\oh}[1]{o\left({#1}\right)}
\newcommand{\Th}[1]{\Theta\left({#1}\right)}

\newcommand{\setN}{\mathbb{N}}
\newcommand{\setZ}{\mathbb{Z}}

\newcommand{\ceil}[1]{\lceil #1 \rceil}
\newcommand{\floor}[1]{\lfloor #1 \rfloor}

\usepackage{microtype}

\usepackage{amsthm}
\usepackage{aliascnt}
\newtheorem{theorem}{Theorem}
\newtheorem{corollary}{Corollary}
\newtheorem{lemma}{Lemma}
\newtheorem{definition}{Definition}

\newcommand{\LF}{\mathsf{LF}}
\newcommand{\A}{\mathsf{A}}
\newcommand{\SA}{\mathsf{SA}}

\newcommand{\SAs}{\SA_\mathsf{s}}

\newcommand{\SAPhi}{\SA_\Phi}

\newcommand{\IPhi}{I_{\Phi}}
\newcommand{\MDS}{\mathcal{M}}

\newcommand{\MLF}{\MDS^{\LF}}
\newcommand{\MPhi}{\MDS^{\Phi}}

\newcommand{\rank}{\mathsf{rank}}
\newcommand{\select}{\mathsf{select}}

\newcommand{\RSLap}{RS_{L'}}
\newcommand{\Ad}{A^d}
\newcommand{\Sarr}{\mathsf{S}}
\newcommand{\PL}{\mathsf{PL}}
\newcommand{\PS}{\mathsf{PS}}
\newcommand{\R}{\mathsf{R}}
\newcommand{\kmerset}{\mathcal{K}}
\newcommand{\ISA}{A^{-1}}
\newcommand{\SR}{\mathsf{SR}}
\newcommand{\CPL}{\mathsf{CPL}}
\newcommand{\SCP}{\mathsf{SCP}}
\newcommand{\LP}{\mathsf{LP}}
\newcommand{\PT}{\mathsf{PT}}
\newcommand{\xcp}{x_\mathsf{cp}}
\newcommand{\xlp}{x_\mathsf{lp}}
\newcommand{\HAd}{H^{\#}_{\Ad}}
\newcommand{\Ts}{\mathcal{T}_\mathsf{s}}
\newcommand{\Tg}{\mathcal{T}_\mathsf{g}}

\usepackage[linesnumbered,noend]{algorithm2e}
\SetFuncSty{}
\SetArgSty{}
\SetKwInOut{Input}{Input}
\SetKwInOut{Output}{Output}
\SetKw{Break}{break}
\SetKw{Continue}{continue}
\SetKw{Or}{or}
\SetKw{And}{and}
\SetKwProg{Fn}{Function}{:}{}

\usepackage{tikz,pgfplots}

\usepgfplotslibrary{groupplots}

\usetikzlibrary{positioning}
\usetikzlibrary{decorations.pathreplacing,calligraphy}
\usetikzlibrary{arrows.meta}
\usetikzlibrary{shapes.geometric}

\tikzset{
    scriptsize/.style={
        font=\scriptsize
    }
}

\pgfplotscreateplotcyclelist{my-colors}{%
    solarized-red, every mark/.append style={solid,scale=0.875}, mark=square \\
    solarized-blue, every mark/.append style={solid,scale=1}, mark=o \\
    solarized-yellow, every mark/.append style={solid,scale=1}, mark=star \\
    solarized-magenta, every mark/.append style={solid,scale=1,semithick}, mark=Mercedes star flipped \\
    solarized-violet, every mark/.append style={solid,scale=1}, mark=x \\
    solarized-green, every mark/.append style={solid,scale=1,semithick}, mark=Mercedes star \\
}

\pgfplotsset{
    compat=1.18,
    major grid style = { thin, dotted, color = black!50 },
    minor grid style = { thin, dotted, color = black!50 },
    grid,
    every tick label/.append style={font=\tiny},
    every axis label/.style={font=\scriptsize},
    every axis/.append style={
        cycle list name=my-colors,
    },
    xlabel near ticks,
    ylabel near ticks,
    legend pos=south east,
    legend cell align=left,
    legend columns=3,
    legend style={
        inner sep=0.25ex,
        outer sep=0,
        column sep=0,
        font=\scriptsize,
        anchor=north west,
        /tikz/every even column/.append style={column sep=3mm,black},
        /tikz/every odd column/.append style={black},
    },
}

\newcommand{\occ}{\text{occ}}
\newcommand{\polylog}{\text{polylog}}

\title{RLZ-$r$ and LZ-End-$r$: Enhancing Move-$r$} 

\author{
    Patrick Dinklage \and
    Johannes Fischer \and
    Lukas Nalbach \and
    Jan Zumbrink
}

\begin{document}

\maketitle

\begin{abstract}
In pattern matching on strings, a locate query asks for an enumeration of all the occurrences of a given pattern in a given text.
The $r$-index~[Gagie et al.,~2018] is a recently presented compressed self index that stores the text and auxiliary information in compressed space.
With some modifications, locate queries can be answered in optimal time [Nishimoto \& Tabei,~2021], which has recently been proven relevant in practice in the form of Move-$r$~[Bertram et al.,~2024].
However, there remains the practical bottleneck of evaluating function $\Phi$ for every occurrence to report.
This motivates enhancing the index by a compressed representation of the suffix array featuring efficient random access, trading off space for faster answering of locate queries [Puglisi \& Zhukova,~2021].

In this work, we build upon this idea considering two suitable compression schemes: Relative Lempel-Ziv~[Kuruppu et al.,~2010], improving the work by Puglisi and Zhukova, and LZ-End~[Kreft \& Navarro,~2010], introducing a different trade-off where compression is better than for Relative Lempel-Ziv at the cost of slower access times.
We enhance both the $r$-index and Move-$r$ by the compressed suffix arrays and evaluate locate query performance in an experiment.

We show that locate queries can be sped up considerably in both the $r$-index and Move-$r$, especially if the queried pattern has many occurrences.
The choice between two different compression schemes offers new trade-offs regarding index size versus query performance.
\end{abstract}

\section*{Disclaimer}
This is the full version of the paper of the same title to be published with the \emph{String Processing and Information Retrieval} (SPIRE) conference in 2025 (CITATION PENDING).
It contains additional details that had to be omitted from the conference paper due to space constraints.

\section{Introduction}
\label{sect:intro}

Searching for occurrences of a pattern in a text or a collection of texts is a ubiquitous problem.
A common use case is to query different patterns against the same text, e.g., picture different users searching for different terms on (a snapshot of) the internet or DNA reads being matched in a genomic database.
This scenario is typically tackled by building an \emph{index} on the text, a data structure that allows for efficient pattern matching queries.
Since the text can be prohibitively large to be indexed plainly, we are very much interested in compressed indexes.
Arguably an important milestone in this area was the invention of the $r$-index by Gagie et al.~\cite{DBLP:conf/soda/GagieNP18} that can be stored in space $\Oh{r}$, where $r$ is the number of runs in the text's Burrows-Wheeler transform -- a well-accepted measure of compressibility.
Augmented by the move data structure by Nishimoto et al.~\cite{DBLP:conf/icalp/NishimotoT21}, pattern matching queries can be answered in optimal time.
Bertram et al.~\cite{DBLP:conf/wea/Bertram0N24} recently implemented this and presented Move-$r$, achieving a very good practical time/space trade-off.

We differentiate between two types of queries:
while \emph{count} queries tell us how often a pattern occurs in the text, \emph{locate} queries ask for an enumeration of all positions at which the pattern occurs.
For this, in the $r$-index, we need to evaluate a function $\Phi$ for every occurrence, which turns out to be the main performance bottleneck in practice.
In independent work (Move-$r$ was not around yet), Puglisi and Zhukova~\cite{DBLP:conf/spire/PuglisiZ20,DBLP:conf/dcc/PuglisiZ21} considered storing a compressed representation of the suffix array alongside the $r$-index that features efficient random access.
For locate queries, we can now directly decode the relevant portion of the suffix array instead of evaluating $\Phi$ for every step.
This resulted in a new trade-off where locate queries could be answered much faster, at the cost of having to store a compressed representation of the suffix array alongside the index.

\subsubsection{Our contributions.}
We transfer the idea of~\cite{DBLP:conf/dcc/PuglisiZ21} and explore enhancing Move-$r$ by a compressed representation of the suffix array with efficient random access, expecting this to be a practical trade-off for speeding up locate queries.
For this, we consider two different compression schemes.
First, like~\cite{DBLP:conf/dcc/PuglisiZ21}, we consider Relative Lempel-Ziv, where we greatly improved the reference construction as well as the parsing procedure.
While the source code of~\cite{DBLP:conf/dcc/PuglisiZ21} remains closed, we publish our reimplementation under an open source license.
Second, we consider LZ-End~\cite{DBLP:conf/dcc/KreftN10}, a different Lempel-Ziv compression scheme that allows for efficient random access.
Here, we give a competitive generalized and simplified algorithm to compute the LZ-End parsing of a string over an integer alphabet based on the (suboptimal) $\Oh{n \lg\lg n}$-time algorithm of Kempa and Kosolobov~\cite{DBLP:conf/esa/KempaK17}.
We also improve Move-$r$ itself by engineering internal rank and select data structure.

We implement different variations of the $r$-index and Move-$r$ and show trade-offs between index size and query performance in our experiments.

\section{Preliminaries}
\label{sect:prelims}

Let $\Sigma \subseteq \setN$ be an integer alphabet and $T \in \Sigma^n$ a \emph{text} over $\Sigma$ of length $n$.
In this work, we are interested in pattern matching queries asking for occurrences of a given pattern $P \in \Sigma^m$ of length $m$ in $T$.
Particularly, we are interested in the queries
(a)~\emph{count}, asking for the number \emph{\occ{}} of occurrences of $P$ in $T$, and
(2)~\emph{locate}, asking for an enumeration of the starting positions of the occurrences.
For some $i \in [1,n]$, we denote by $T[i]$ the $i$-th character in $T$.
Given additionally $j \in [i, n]$, we denote by $T[i ~..~ j]$ the substring $T[i]~T[i+1] \cdots T[j-1]~T[j]$, juxtaposition meaning concatenation.
The aforementioned queries are formally defined as $\text{locate}(T, P) = \{ i \in [1, n-m+1] ~|~ T[i ~..~ i+m-1] = P \}$ and $\text{count}(T, P) = |\text{locate}(T, P)|$.
We argue about running times in the word RAM model, where we can do operations on words of length $\omega = \Th{\lg n}$ bits in constant time.
Unless explicitly stated otherwise, logarithms are given as base-two.

\subsection{Lempel-Ziv Parsings and Random Access}
\label{sect:lz}

Lempel-Ziv (LZ) parsings factorize a text $T$ into $z \leq n$ \emph{phrases} $f_1, \dots, f_z$ such that their concatenation $f_1 \cdots f_z = T$.
They form a family of dictionary compressors, with arguably the most popular representative being Lempel-Ziv~77 (LZ77)~\cite{DBLP:journals/tit/ZivL77}.
There, we define the phrase $f_i$ (for $i \in [1,z]$) as either
(1)~a new symbol that does not occur in $T[1 ~..~ |f_1 \cdots f_{i-1}|]$, or
(2)~the longest prefix of $T[|f_1 \cdots f_{i-1}| + 1 ~..~ n]$ that has an occurrence in $T$ starting at a position $\leq |f_1 \cdots f_{i-1}|$.
In the second case, we can encode the phrase as a \emph{reference} to a previous occurrence, which can potentially be stored in less bits than storing the phrase explicitly.
Even if we relax the definition of referencing phrases, it is typically $z \ll n$ if $T$ is repetitive.
Thus, LZ parsings are a popular choice for compressing $T$ and are used in myriad everyday utilities such as \emph{gzip}.

\subsubsection{Random access.}

We are interested in efficient random access on $T$ in its compressed form, i.e., we wish to extract a substring $T[x ~..~ x+\ell]$ for some $x \in [1,n]$ and $\ell \geq 0$ without decoding substantial portions of $T$.
Regarding just the character $T[x]$, it is contained in the phrase $f_i$ for $i = \min\{ i \in [1,z] ~\vert~ |f_1 \cdots f_i| \geq x\}$.
We say that phrase $f_i$ \emph{covers} position $x$.
We can store the set $E = \{ |f_1 \cdots f_j| ~\vert~ j \in [1,z] \}$ (the end positions of the phrases) in $z\ceil{\lg n}$ bits (i.e., in space $\Oh{z}$) and compute $i$ in time $\Oh{\lg z}$ using binary search, or use a static successor data structure to compute $i$ in time $\Oh{\lg\lg(n/z)}$~\cite{DBLP:conf/stoc/PatrascuT06}.
\footnote{
    Alternatively, we can build the characteristic bit vector $B_E$ of $n$ bits where the $j$-th bit is set iff $j \in E$.
    Since exactly $z$ bits are set in $B_E$, we can build a data structure of size $\ceil{\lg\binom{z}{m}} + \oh{n} + \Oh{\lg\lg z}$ bits that supports constant-time rank and select queries on $B_E$~\cite{DBLP:conf/soda/RamanRR02}.
    With this, we can then also compute $i$ in constant time.
    In this work, however, our aim is to focus on compressed space $\Oh{z}$.
}
An inherent disadvantage of the classic LZ77 scheme defined above, albeit achieving very good compression in practice, is that random access cannot be done efficiently.
Phrases may refer to arbitrary prior positions in $T$, and thus to decode $T[x]$, we may have to decode all phrases $f_1, \dots, f_i$ up to (a prefix of) the phrase that covers $x$.
We now look at two variants that resolve this issue at the cost of worse compression.

\subsubsection{LZ-End.}
\label{sect:lzend}

Kreft~and~Navarro introduced the scheme \emph{LZ-End}~\cite{DBLP:conf/dcc/KreftN10}.
Here, each phrase $f_i$ is represented as a triple $(j, \ell, \alpha)$, where $j < i$ is the \emph{source phrase}, $\ell \geq 0$ is the \emph{copy length} and $\alpha \in \Sigma$ is a character such that $f_i = T[|f_1 \cdots f_j| - \ell + 1 ~..~ |f_1 \cdots f_j|] ~ \alpha$ for maximal possible $\ell$ and there is no $k < i$ such that $f_i$ is a suffix of $T[1 ~..~ |f_1 \cdots f_k|]$.
We allow $f_0 := \epsilon$ as a valid source phrase such that the above is well-defined.
Intuitively, $f_i$ extends the length-$\ell$ suffix of $T[1 ~..~ |f_1 \cdots f_j|]$ by a new character $\alpha$ and $j$ is picked greedily such that $\ell$ is maximized.

Since each phrase adds exactly one character to a previously occurring substring, the end position of which is explicitly stated in the encoding triple, we can decode $T[x ~..~ x+\ell]$ in time $\Oh{h+\ell}$ once we know the phrase that covers position $x+\ell$.
Here, $h$ is the length of the longest phrase. 
This gives us total random access time $\Oh{\lg\lg(n/z_{\text{end}}) + h + \ell}$, where $z_\text{end} = \Oh{z \lg^2 n}$ is the number of LZ-End phrases of $T$~\cite{DBLP:conf/soda/KempaS22}.

When computing the parsing, we can artificially constrain $h$ to obtain parameterized random access time at the cost of introducing only at most $\Oh{(n \lg n)/h}$ additional phrases~\cite{DBLP:conf/soda/KempaS22} ($\Oh{n/h}$ of which come naturally from a phrase length restriction of $h$).
\footnote{
    We thank the anonymous reviewer for the insightful comments on the behaviour of LZ-End under phrase length restrictions.
}

\subsubsection{Relative Lempel-Ziv.}
\label{sect:rlz}

Kuruppu~et~al. proposed a variant of Lempel-Ziv parsings where we do not refer to earlier parts of $T$ itself, but instead to a given reference $R \in \Sigma^*$ \cite{DBLP:conf/spire/KuruppuPZ10}.
This is useful especially in scenarios where we want to store a collection of texts that are highly similar (e.g., genomic sequences from the same species).
Formally, the phrase $f_i$ is the longest prefix of $T[|f_1 \cdots f_{i-1}| + 1 ~..~ n]$ that occurs in $R$, or a single character that does not occur in $R$.
This scheme is referred to as \emph{Relative Lempel-Ziv} (RLZ).

To decode $T[x ~..~ \ell]$, we can directly access the substring in $R$ that the phrase covering $x$ refers to, which can be done in total time $\Oh{\lg\lg(n/z_R) + \ell}$, where $z_R$ is the number of RLZ phrases of $T$ for reference $R$.
The compression depends on how well $R$ represents $T$, and $R$ must be stored alongside the compressed form of $T$ in order to be able to decode $T$.

\subsection{Suffix Arrays, Burrows-Wheeler Transform and Compression}
\label{sect:csa}

In the \emph{suffix array} $A$ of $T$, we store the starting positions of the suffixes of $T$ in their lexicographical order~\cite{DBLP:journals/siamcomp/ManberM93}.
This ordering causes suffixes that begin with equal prefixes to be grouped in consecutive intervals.
A text book algorithm to answer count queries in time $\Oh{m \lg n}$ finds the interval $[b,e] \subseteq [1,n]$ of $A$ that contains all (and only the) suffixes of $T$ beginning with $P$, the query time stemming from binary searches for $b$ and $e$, respectively.
To answer locate, we simply need to enumerate $A[b ~..~ e]$.
We can store $A$ in $n\ceil{\lg n}$ bits of space and construct it in time $\Oh{n}$~\cite{DBLP:journals/tc/NongZC11}.

The \emph{Burrows-Wheeler transform} (BWT) of $T$ is a reversible transform of $T$ defined as $L[i] := T[A[i]-1]$ (or $L[i] := T[n]$ if $A[i] = 1$)~\cite{bwt}.
The BWT of repetitive texts tends to contain long equal-letter runs, which can be exploited by run-length compression.
We denote by $r$ the number of these runs.

\subsubsection{Compressed Differential Suffix Arrays.}

In practice, storing $A$ plainly is prohibitive for large $T$.
Even though it is a permutation over $[1,n]$ and thus not inherently compressible, different ways to compress $A$ have been shown (we refer to \cite{DBLP:journals/tcs/Grossi11} for an overview).
In this work, we focus on compressing the \emph{differential suffix array} $A^d \in \setZ^n$, where $A^d[1] := A[1]$ and $A^d[i] := A[i]-A[i-1]$ for $i \in [2, n]$.

Gonz\'{a}lez~et~al. first exploited the interesting property that the number of distinct values in $A^d$ is bounded by the number $r$ of BWT runs and repetitiveness in $T$ implies repetitiveness in $A^d$~\cite{DBLP:journals/jea/GonzalezNF14}.
In this work, we are interested in the approach by Puglisi and Zhukova~\cite{DBLP:conf/dcc/PuglisiZ21}, who instead considered RLZ to compress $A^d$.
They describe a strategy to extract $R$ from $A^d$ by selecting segments based on the frequencies of representative substrings, and show that this outperforms using a random sample of $A^d$ (for which bounds on the expected compression have been shown~\cite{DBLP:conf/spire/GagiePV16}).
We denote by $\hat{z}_R$ the number of RLZ phrases of $A^d$ computed this way.

For random access on $A$, we want to avoid having to compute $A[x] = \sum_{i=1}^x A^d[i]$ for some $x \in [1,n]$ in time $\Oh{n}$.
Rather, we create a sample $A'$ that contains a subsequence of $A$.
Let $y$ be the greatest sampled position $\leq x$, then we can compute $A[x] = A'[y] + \sum_{i=y+1}^x A^d[i]$ in time $\Oh{\delta}$, where $\delta$ is the maximum distance between any position and the previous sample.

For example, in \cite{DBLP:conf/dcc/PuglisiZ21}, we take a sample of $A$ for every RLZ phrase.
This gives us $\delta < h$ and using a static successor data structure of size $\Oh{\hat{z}_R}$ (since $|A'| = \hat{z}_R$), random access is possible in time $\Oh{\lg\lg(n/\hat{z}_R) + h}$, where $h$ is the length of the longest RLZ phrase.

\subsection{(Move-)$r$-Index}
\label{sect:rindex}

The $r$-index is a recent advancement in compressed data structures for pattern matching~\cite{DBLP:conf/soda/GagieNP18} that is also highly relevant in practice.
It is a self-index that encodes the BWT of $T$ and auxiliary data structures in $\Oh{r}$ space.
Using the move data structure of~\cite{DBLP:conf/icalp/NishimotoT21} and assuming $|\Sigma| = \Oh{\polylog~n}$, we obtain optimal $\Oh{m}$ time for count and optimal additional time $\Oh{\occ{}}$ for locate queries, where $\occ{}$ is the number of occurrences of the search pattern.

\section{LZ-End Compression of Suffix Arrays}
\label{sect:lzendsa}

Following the idea of \cite{DBLP:conf/dcc/PuglisiZ21} to apply LZ compression on the differential suffix array, we explore its compression using LZ-End.
To give an intuition as to why this may be fruitful, LZ-End
(1)~allows for efficient random access on the compressed input and
(2)~achieves competitive compression in practice.
We first show the following for compressing any integer sequence $A$.
\begin{theorem}
\label{thm:lzend-array}

Let $A \in [1,n]^n$ be an integer sequence.
In time and space $\Oh{n}$, we can construct a data structure of size $\Oh{\hat{z}_\text{end}}$ such that for $x \in [1,n]$ and $\ell \geq 0$, we can reconstruct $A[x ~..~ x+\ell]$ in time $\Oh{\lg\lg(n/\hat{z}_\text{end}) + h + \ell}$, where $\hat{z}_\text{end}$ is the number of LZ-End phrases of the differential representation $A^d$ of $A$ and $h$ the length of the longest phrase.
\begin{proof}
The differential representation $A^d$ of $A$ can be computed in time and space $\Oh{n}$ and by~\cite{DBLP:conf/esa/KempaK17}, the same holds for the LZ-End parsing of $A^d$.
We represent the triples defining the parsing as three arrays:
\begin{enumerate}
\item the array \emph{src}, where the $i$-th entry contains the number $\in [1,i-1]$ of the source phrase that $f_i$ refers to,
\item the array \emph{end}, where the $i$-th entry contains the position $|f_1 \cdots f_i| \in [1,n]$ at which phrase $f_i$ ends in $A^d$, and
\item the array \emph{ext}, where the $i$-th entry contains the value $\in [-n,n]$ from $A^d$ that extends the suffix of $A^d[1 ~..~ |f_1 \cdots f_{\textbf{src}[i]}|]$.
\end{enumerate}
Each array can be stored in space $\Oh{\hat{z}_\text{end}}$.
We also build a static successor data structure over \emph{end} that allows for successor queries in time $\Oh{\lg\lg(n/\hat{z}_\text{end})}$, which we can do in time and space $\Oh{\hat{z}_\text{end}}$.
Finally, in time at most $\Oh{\hat{z}_\text{end}}$, we sample the $\hat{z}_\text{end}$ values from $A$ at the positions stored in \emph{end} in a new array $A'$ and store them also in space $\Oh{\hat{z}_\text{end}}$.

Given $x \in [1,n]$ and $\ell \geq 0$, we decode $A[x ~..~ x+\ell]$ as follows:
we first extract the range $A^d[x ~..~ x+\ell]$ from the LZ-End parsing in time $\Oh{\lg\lg(n/\hat{z}_\text{end}) + h + \ell}$ using the extraction algorithm from~\cite{DBLP:conf/dcc/KreftN10} (the length of a phrase $f_i$ can trivially be computed in constant time as $|f_i| = \text{end}[i] - \text{end}[i-1]$).
Using the successor data structure, we can find in time $\Oh{\lg\lg(n/\hat{z}_\text{end})}$ the position of a relevant sample of $A$ that is stored in $A'$.
Then, in time at most $\Oh{h + \ell}$, we accumulate the relevant differential values from $A^d[x ~..~ x+\ell]$ to reconstruct $A[x ~..~ x+\ell]$.
\end{proof}
\end{theorem}

\begin{corollary}
\label{corr:lzendsa}
Let $T \in \Sigma^n$ be a string of length $n$.
In time and space $\Oh{n}$, we can construct a data structure of size $\Oh{\hat{z}_\text{end}}$ such that for $x \in [1,n]$ and $\ell \geq 0$, we can compute the suffix array interval $A[x ~..~ x +\ell]$ in time $\Oh{\lg\lg(n/\hat{z}_\text{end}) + h + \ell}$, where $\hat{z}_\text{end}$ is the number of LZ-End phrases of the differential representation $A^d$ of the suffix array $A$ of $T$ and $h$ is the length of the longest phrase.
\end{corollary}

\subsection{Practical LZ-End Parsing}
\label{sect:lzend-s}

To implement the computation of LZ-end parsings, we adopt and modify the algorithm by Kempa and Kosolobov~\cite{DBLP:conf/esa/KempaK17} that does so in time $\Oh{n \lg\lg n}$ in a left-to-right scan of the text $T$ of length $n$ (a linear-time algorithm exists~\cite{DBLP:conf/esa/KempaK17}, but we conjecture it to be hardly practical).
At its core lies a dynamic predecessor/successor data structure $M$ that marks the lexicographic ranks of suffixes of the reverse input $\overleftarrow{T}$ at which already computed phrases end.
In the following, we briefly describe our modifications and refer to Appendix~\ref{appendix:lzend-s-ext} for details.

First, we make $M$ associative, so that at each marked suffix, we also store the number of the phrase that ends at the suffix.
This removes a level of indirection and even allows us to completely discard the suffix array after initialization.

Second, say that $f_1 \cdots f_x$ is the LZ-End parsing computed thus far for $x < z_{\text{end}}$ and we now read the next character $\alpha \in \Sigma$.
One of the possible cases for the next LZ-end phrase is merging the phrases $f_{x-1}$ and $f_x$ to a new phrase $f_{x-1} f_x \alpha$.
In this case, clearly, we cannot use phrase $x-1$ as a source phrase to copy from.
In the original algorithm, we temporarily unmark phrase $x-1$ in $M$ to guarantee that it is not reported as a possible a source phrase, and then afterwards mark it back in $M$ unless it was merged with phrase $x$.
We reduce the overall number of updates on $M$ by performing an additional predecessor or successor query in case $x-1$ is reported, and only ever unmark a phrase in $M$ when it is actually removed.

Third, the parsing is computed processing $T$ left to right, but suffixes of the reverse $\overleftarrow{T}$ are considered.
We save arithmetic computations by using a variant $A^{\leftarrow 1}$ of the inverse suffix array of $T$ that is defined as $A^{\leftarrow 1}[A[n-i-1]] := i$.
Then, it is $A^{\leftarrow 1}[i] = A^{-1}[n-i]$ and we read $A^{\leftarrow 1}$ left to right as we process $T$.

Finally, our implementation of the algorithm is written in a way that $\Sigma$ may be an arbitrary integer alphabet such that, e.g., we can compute the parsing for a differential suffix array.

We set the maximum phrase length to $h := 2^{13}$, giving us the best balance in preliminary experiments: choosing lower $h$ notably increased the number of phrases, whereas higher---or unbounded---$h$ had an increasingly negative impact on access performance.
Furthermore, we store the array \emph{end} of end positions plainly using $z\ceil{\lg n}$ bits and use a simple $\Oh{\lg z}$-time binary search with no auxiliary data structure to find the phrase covering a position in question.
Preliminary experiments have shown that despite its simplicity, this approach is the fastest given the relatively low number $z$.

In Appendix~\ref{appendix:lzend-s-ext}, we present results of experiments showing that our implementation of LZ-End is competitive with the \emph{lz-end-toolkit} and faster on general (non-highly repetitive) inputs.

\section{Improved RLZ Compression of Suffix Arrays}
\label{sect:rlzsa}

Puglisi and Zhukova~\cite{DBLP:conf/dcc/PuglisiZ21} considered compressing the differential suffix array using Relative Lempel-Ziv (henceforth referred to as RLZSA).
However, their source code remains closed.
With the aim of reproducing their results for further research, we reimplemented RLZSA as described there and in Zhukova's doctoral thesis~\cite{DBLP:phd/basesearch/Zhukova24} to the best of our capabilities.
In this process, we found several clues as to how to improve upon their work.
We summarize our improvements here and refer to Appendix~\ref{appendix:rlzsa} for an in-depth description of the individual steps.

First, we shrink the overall representation by separating data pertaining to literal phrases (encoding explicitly a character $\alpha \in \Sigma$) or copying RLZ phrases (encoding the source position in $R$ and the number of characters to copy).
We store a bit vector of length $z_R$ with support for rank/select that allows us to identify the type of each phrase and use rank and select queries to access its data in separate arrays.
This allows us to drop the requirement that each copying phrase needs to be preceded by a literal phrase, which also improves RLZ compression.
Using the new representation, we can reduce the time for randomly accessing a suffix array interval $A[b ~..~ e]$ from $\Oh{|e-b| + h + \lg(z_R/a) + a}$ to $\Oh{|e-b| + \lg(na/z_R) + a}$ (see Appendix~\ref{appendix:rlzsa_data_structures}~and~\ref{appendix:rlzsa_queries}),
where $a \geq 1$ is an integer sampling parameter (for sampling phrase starting positions) and $h$ is the length of the longest RLZ phrase.

We then proceed to improve upon the construction of RLZSA.
By using Big-BWT~\cite{pfp}, we can make the construction of $A^d$ semi-external, reducing the memory usage from $\Oh{n}$ to $\Oh{|\mathsf{PFP}|}$, where $\mathsf{PFP}$ is a prefix free parsing of $T$.
By allowing the selection of arbitrary segments from $A^d$ (instead of partitioning $A^d$ and only allowing aligned segments) and by setting the considered $k$-mer length to $k := 1$,
we can reduce the time and space required for reference construction from $\Oh{n}$ to $\Oh{r^{1-\epsilon} n^\epsilon}$ and $\Oh{r}$, respectively, where $\epsilon \in [0,1]$ is a parameter (see Appendix~\ref{appendix:rlzsa_ref_constr}).
The new segment selection strategy also improves the reference quality, leading to better compression (as shown later in Table~\ref{tab:inputs}).

Finally, we speed up and reduce the memory usage for computing the RLZ parsing of $A^d$ for the computed reference $R$ by replacing the FM-index by Move-$r$ over $\overleftarrow{R}$ using an optimized rank/select data structure for large alphabets from Appendix~\ref{appendix:rank_select_large}.

We set the size of the RLZ reference $|R| := \min(5.2r, n/3)$, which gave us the best results overall in preliminary experiments.
RLZ phrases are limited to maximum length $h := 2^{16}$, which allows storing their length in 16-bit integers.

\section{Applications to the (Move)-$r$-index}
\label{sect:move-r-improvements}

The main bottleneck when answering locate queries using the $r$-index in practice are the applications of the function $\Phi$ required to enumerate occurrences.

There have been at least two different approaches to resolve this, both of which have been shown to be relevant in practice:
Puglisi and Zhukova store the RLZ-compressed differential suffix array next to the $r$-index, which allows for up to two orders of magnitude faster locate queries (with many occurrences) at the cost of using 2--13 times as much memory~\cite{DBLP:conf/dcc/PuglisiZ21}.
Bertram et al., on the other hand, implement the move data structure by \cite{DBLP:conf/icalp/NishimotoT21}, speeding up queries (both count and locate) by an order of magnitude while only doubling the required space~\cite{DBLP:conf/wea/Bertram0N24}.

We propose variations and combinations of the above and evaluate them in our experiments (Section~\ref{sect:experiments}).
Namely, we explore storing an LZ-End-compressed differential suffix array next to the $r$-index to find out whether we can obtain a trade-off similar to~\cite{DBLP:conf/dcc/PuglisiZ21}.
Furthermore, albeit much faster than in the original $r$-index, the $\Phi$ steps for enumerating occurrences remain a bottleneck for the locate queries of~\cite{DBLP:conf/wea/Bertram0N24}, each causing cache misses.
We thus also consider storing either compressed differential suffix array next to Move-$r$.
The space requirement for the whole index increases by at most a polylogarithmic factor, which follows from known bounds on LZ-End and differential suffix arrays~\cite{r-index-ext,DBLP:conf/soda/KempaS22}
\footnote{
    We again thank the anonymous reviewer for their notes regarding the additional required space.
}.

The $r$-index (as well as Move-$r$) already maintains a sampling $A'_r$ of the suffix array at the boundary of every BWT run, i.e., $|A'_r| = r$.
This creates redundancy regarding the sampling $A'$ of suffix array values at LZ phrase end positions proposed by \cite{DBLP:conf/dcc/PuglisiZ21} and in the proof of Theorem~\ref{thm:lzend-array}.
For reconstructing a suffix array value using $A_d$, we can as well use $A'_r$.
This worsens the worst-case access time to $\Oh{n}$, because there is no general bound on the length of a BWT run.
In practice, however, the average length of a BWT run is reasonably short even for repetitive inputs (see, e.g., column $\floor{n/r}$ in Table~\ref{tab:inputs}).

Alternatively, one could consider replacing the sampling $A'_r$ by $A'$ in the $r$-index or Move-$r$.
However, then, to retrieve the suffix array value at the end of a run, we must spend up to $\Oh{\lg\lg(n/\hat{z}) + h}$ extra time time for random access, which would worsen the performance of locate queries, conflicting with our motivations.
It would also increase the index size in practice as empirically, it holds that $\hat{z} > r$.
Therefore, we do not further consider this sampling method.

\section{Experiments}
\label{sect:experiments}

In our experiments, we evaluate the construction and locate query performance of the following variations of the $r$-index and Move-$r$:
\begin{itemize}
\item \texttt{r-index}      -- the original $r$-index of~\cite{DBLP:conf/soda/GagieNP18},
\item \texttt{r-rlz}        -- the $r$-index plus the RLZ-compressed differential suffix array,
\item \texttt{r-lzend}      -- the $r$-index plus the LZ-End-compressed differential suffix array,
\item \texttt{move-r}       -- the Move-$r$ index of~\cite{DBLP:conf/wea/Bertram0N24} (with improved internal rank/select),
\item \texttt{move-r-rlz}   -- Move-$r$ plus the RLZ-compressed differential suffix array, and
\item \texttt{move-r-lzend} -- Move-$r$ plus the LZ-End-compressed differential suffix array.
\end{itemize}
Note that only \texttt{move-r-rlz} contains the improved RLZSA construction that we described in Section~\ref{sect:rlzsa}, whereas \texttt{r-rlz} is based on a reimplementation of~\cite{DBLP:conf/dcc/PuglisiZ21} described in Appendix~\ref{appendix:rlzsa}.
We do this to better argue about our improvements.
However, we use Big-BWT~\cite{pfp} for all variants to compute suffix arrays.
As mentioned in the list above, we also applied improvements to Move-$r$ itself by engineering a new rank/select data structure tailored specifically for its internal queries.
We refer to Appendix~\ref{appendix:rank_select} for details.

We implemented all index variants in {C++20} and make the source code publicly available\footnote{Our source code: \url{https://github.com/LukasNalbach/Move-r}.}.
We compiled using the GCC 13.3.0 compiler with flags set for highest optimization (\texttt{-march=native -DNDEBUG -Ofast}).

Table~\ref{tab:inputs} lists the input texts that we considered in our experiments alongside relevant statistics.
The texts \texttt{einstein} and \texttt{english} are part of the Pizza\&Chili Corpus\footnote{Pizza\&Chili Corpus: \url{https://pizzachili.dcc.uchile.cl/}},
whereas \texttt{dewiki} is a highly repetitive text manually constructed from German Wikipedia entries.
From the \emph{National Center for Biotechnology Information}\footnote{NCBI: \url{https://www.ncbi.nlm.nih.gov/}} (NCBI) database, we constructed \texttt{chr19}, consisting of concatenated human chromosome 19 haplotypes, and \texttt{sars2}, a collection of Sars-Cov-2 genomes.
From all text files, we erased all zero bytes.

For each text, we generated two sets of query patterns (hence two lines per file in Table~\ref{tab:inputs}) using our tool \texttt{move-r-patterns} (also included in our source code repository).
The sets differ in the pattern length $m$, as well as the average number $\overline{occ}$ of occurrences in the respective text.
We chose the patterns in the first set such that $\overline{\occ{}} \approx m$.
This implies that when locating those patterns, we measure a blend of backward-search and suffix array extraction.
The performance of counting queries was measured against this set.
The patterns in the second set were chosen such that $\overline{\occ{}} \approx 10^5m$.
When locating these, we measure mostly suffix array extraction, which is a particularly relevant measure for our experiments.

We note that we do not consider $m \gg \occ{}$, i.e., where we have long patterns with few occurrences.
Measurements in this realm would essentially measure the performance of count queries, which has already been done for both the $r$-index as well as Move-$r$~\cite{DBLP:conf/wea/Bertram0N24}.

All experiments were done on a Ubuntu 24.04 system with two AMD EPYC 7452 CPUs (32/64x 2.35-3.35GHz, 2/16/128MB L1/2/3 cache) and 1TB of RAM (3200 MT/s DDR4).

\begin{table}[t]
\caption{
    The input files for our experiments.
    For each input, we give the size $n$, the size $|\Sigma|$ of the alphabet and the compression ratios $n / r$, $n / \hat{z}_R$ and $n / \hat{z}_\text{end}$ (higher values mean more repetitive).
    Here, $\hat{z}_R$ is the number of RLZ phrases of $A^d$ following the construction of~\cite{DBLP:conf/dcc/PuglisiZ21}, whereas $\hat{z}_{R'}$ refers to our improved construction from Section~\ref{sect:rlzsa}.
    As in Section~\ref{sect:lzend-s}, $\hat{z}_\text{end}$ denotes the number of LZ-End phrases of $A^d$.
    By $N$, we denote the number of queried patterns, by $m$ the pattern length and by $\overline{\occ{}}$ the average number of occurrences of the patterns.
    Per input, the first line indicates $N$, $m$ and $\overline{\occ{}}$ for $m \approx \overline{\occ{}}$, the second line for $m \ll \overline{\occ{}}$.
}
\begin{center}
\begin{small}
\begin{tabular}{l | r r | r r r r | r r r }
text                & $n$ [GB]    & $|\Sigma|$ & $\floor{n/r}$ & $\floor{n / \hat{z}_R}$ & $\floor{n / \hat{z}_{R'}}$ & $\floor{n / \hat{z}_\text{end}}$ & $N$ & $m$ & $\overline{\occ{}}$ \\
\hline
einstein            & \num{0.47}  & \num{140} & \num{1611} & \num{118} & \num{183} & \num{1081} & \num{100000} & \num{800} & \num{736} \\
                    & & & & & & &                                                         \num{10000} & \num{7} & \num{72644} \\
sars2               & \num{10.00} & \num{80} & \num{548} & \num{60} & \num{61} & \num{336} & \num{3000} & \num{2700} & \num{2745} \\
                    & & & & & & &                                                         \num{100} & \num{24} & \num{178948} \\
dewiki              & \num{10.00} & \num{207} & \num{377} & \num{122} & \num{146} & \num{306} & \num{100000} & \num{300} & \num{323} \\
                    & & & & & & &                                                         \num{1000} & \num{9} & \num{76372} \\
chr19               & \num{10.00} & \num{53} & \num{46} & \num{12} & \num{25} & \num{34} & \num{1000} & \num{25000} & \num{19531} \\
                    & & & & & & &                                                          \num{1000} &  \num{100} &\num{107991} \\
english             & \num{2.21}  & \num{240} & \num{3} & \num{2} & \num{4} & \num{3} & \num{500000} & \num{35} & \num{37} \\
                    & & & & & & &                                                         \num{300} & \num{7} & \num{91964} \\
\end{tabular}
\end{small}
\end{center}
\label{tab:inputs}
\end{table}

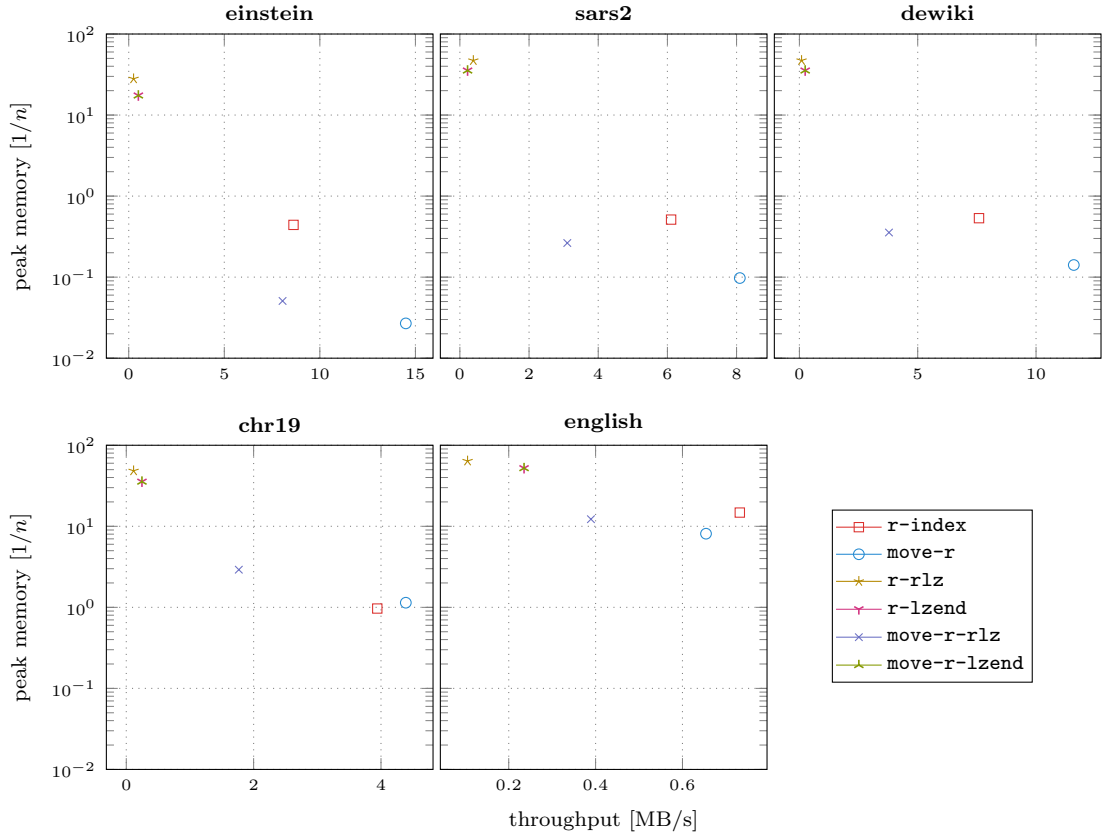
\begin{figure}[t]
\begin{center}
\begin{tikzpicture}

    \begin{groupplot} [
        width=0.4\textwidth,
        height=0.4\textwidth,
        group style={
            ylabels at=edge left,
            group size=3 by 2,
            vertical sep=3em,
            horizontal sep=0.25em,
        },
        title style={
            font=\scriptsize,
            yshift=-1ex,
        },
        legend columns=1,
        legend style={
            at={(1.2, 0.8)},
        }
    ]
    
    
    \nextgroupplot[
        title=\textbf{einstein},
        ylabel={peak memory [$1/n$]},
        ymode=log,
        ymin=0.01,
        ymax=100,
        ytick={0.01, 0.1, 1, 10, 100},
    ]
    \addplot coordinates { (8.62136,0.441342) };
    \addlegendentry{\texttt{r-index}};
    \addplot coordinates { (14.4979,0.0268693) };
    \addlegendentry{\texttt{move-r}};
    \addplot coordinates { (0.253978,28.0382) };
    \addlegendentry{\texttt{r-rlz}};
    \addplot coordinates { (0.504068,17.4333) };
    \addlegendentry{\texttt{r-lzend}};
    \addplot coordinates { (8.04933,0.0508317) };
    \addlegendentry{\texttt{move-r-rlz}};
    \addplot coordinates { (0.504068,17.4333) };
    \addlegendentry{\texttt{move-r-lzend}};
    \legend{}
    
    
    \nextgroupplot[
        title=\textbf{sars2},
        ymode=log,
        ymin=0.01,
        ymax=100,
        yticklabels={},
        ytick={0.01, 0.1, 1, 10, 100},
    ]
    \addplot coordinates { (6.10718,0.512905) };
    \addlegendentry{\texttt{r-index}};
    \addplot coordinates { (8.09363,0.0974567) };
    \addlegendentry{\texttt{move-r}};
    \addplot coordinates { (0.389304,47.15) };
    \addlegendentry{\texttt{r-rlz}};
    \addplot coordinates { (0.222694,35.4234) };
    \addlegendentry{\texttt{r-lzend}};
    \addplot coordinates { (3.10612,0.263709) };
    \addlegendentry{\texttt{move-r-rlz}};
    \addplot coordinates { (0.222694,35.4234) };
    \addlegendentry{\texttt{move-r-lzend}};
    \legend{}
        
    
    \nextgroupplot[
        title=\textbf{dewiki},
        ymode=log,
        ymin=0.01,
        ymax=100,
        yticklabels={},
        ytick={0.01, 0.1, 1, 10, 100},
    ]
    \addplot coordinates { (7.58965,0.533997) };
    \addlegendentry{\texttt{r-index}};
    \addplot coordinates { (11.5951,0.14117) };
    \addlegendentry{\texttt{move-r}};
    \addplot coordinates { (0.0921261,47.1819) };
    \addlegendentry{\texttt{r-rlz}};
    \addplot coordinates { (0.244896,35.4305) };
    \addlegendentry{\texttt{r-lzend}};
    \addplot coordinates { (3.78536,0.355324) };
    \addlegendentry{\texttt{move-r-rlz}};
    \addplot coordinates { (0.244896,35.4305) };
    \addlegendentry{\texttt{move-r-lzend}};
    \legend{}

    
    \nextgroupplot[
        title=\textbf{chr19},
        ylabel={peak memory [$1/n$]},
        ymode=log,
        ymin=0.01,
        ymax=100,
        ytick={0.01, 0.1, 1, 10, 100},
    ]
    \addplot coordinates { (3.94176,0.966607) };
    \addlegendentry{\texttt{r-index}};
    \addplot coordinates { (4.38977,1.1407) };
    \addlegendentry{\texttt{move-r}};
    \addplot coordinates { (0.114888,48.4334) };
    \addlegendentry{\texttt{r-rlz}};
    \addplot coordinates { (0.246281,35.5497) };
    \addlegendentry{\texttt{r-lzend}};
    \addplot coordinates { (1.76728,2.91534) };
    \addlegendentry{\texttt{move-r-rlz}};
    \addplot coordinates { (0.246281,35.5497) };
    \addlegendentry{\texttt{move-r-lzend}};
    \legend{}

    
    \nextgroupplot[
        title=\textbf{english},
        xlabel={throughput [MB/s]},
        ymode=log,
        ymin=0.01,
        ymax=100,
        yticklabels={},
        ytick={0.01, 0.1, 1, 10, 100},
    ]
    \addplot coordinates { (0.732028,14.772) };
    \addlegendentry{\texttt{r-index}};
    \addplot coordinates { (0.654212,8.10482) };
    \addlegendentry{\texttt{move-r}};
    \addplot coordinates { (0.105219,63.9915) };
    \addlegendentry{\texttt{r-rlz}};
    \addplot coordinates { (0.235273,52.1602) };
    \addlegendentry{\texttt{r-lzend}};
    \addplot coordinates { (0.389263,12.2308) };
    \addlegendentry{\texttt{move-r-rlz}};
    \addplot coordinates { (0.235273,52.1602) };
    \addlegendentry{\texttt{move-r-lzend}};

    \end{groupplot}
\end{tikzpicture}
\end{center}
\caption{
    Construction time versus peak memory usage (in bytes per input character) of our implemented index data structures for the given inputs.
	Memory usage is given on a logarithmic scale in order to highlight the marginal differences between \texttt{r-index}, \texttt{move-r} and \texttt{move-r-rlz}.
	Data points for \texttt{r-lzend} and \texttt{move-r-lzend} do, in fact, overlap nearly precisely.
}
\label{fig:results-construction}
\end{figure}

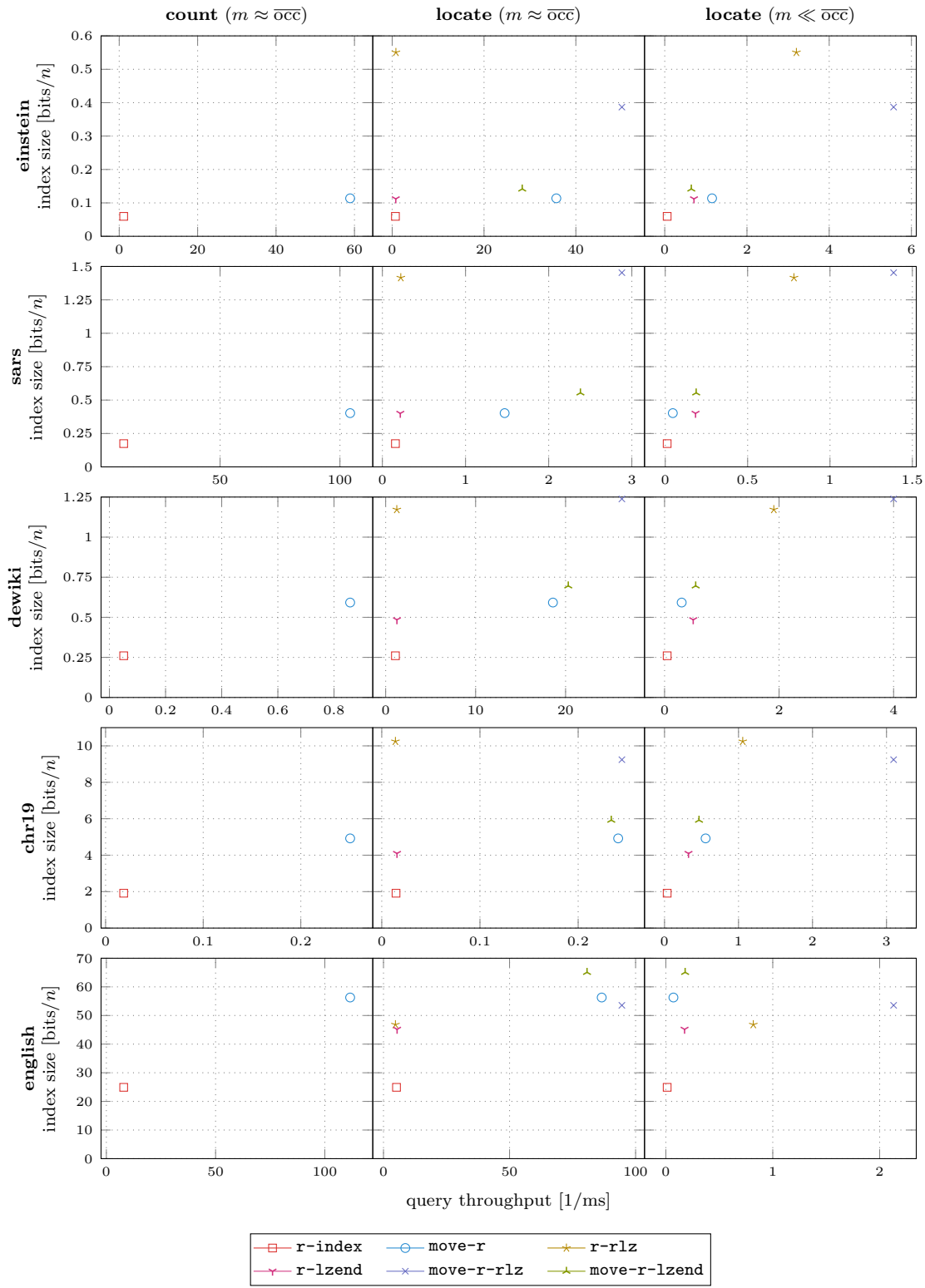
\begin{figure}[p]
\begin{center}
\begin{tikzpicture}

    \begin{groupplot} [
        width=0.4\textwidth,
        height=0.325\textwidth,
        group style={
            ylabels at=edge left,
            group size=3 by 5,
            horizontal sep=0em,
            vertical sep=1.25em,
        },
        title style={
            font=\scriptsize,
            yshift=-1ex,
        },
        legend style={
            at={(-0.45, -0.375)},
        }
    ]
    
    
    \nextgroupplot[
        title=\textbf{count} ($m \approx \overline{\occ{}}$),
        ylabel style={align=center},
        ylabel=\textbf{einstein}\\{index size [bits/$n$]},
        ymin=0,
        ymax=0.6,
        ytick distance=0.1,
    ]
    \addplot coordinates { (1.14348,0.0596715) };
    \addlegendentry{\texttt{r-index}};
    \addplot coordinates { (58.8946,0.113669) };
    \addlegendentry{\texttt{move-r}};
    \legend{}
    
    \nextgroupplot[
        title=\textbf{locate} ($m \approx \overline{\occ{}}$),
        yticklabels={},
        ymin=0,
        ymax=0.6,
        ytick distance=0.1,
    ]
    \addplot coordinates { (0.706795,0.0596715) };
    \addlegendentry{\texttt{r-index}};
    \addplot coordinates { (35.7452,0.113669) };
    \addlegendentry{\texttt{move-r}};
    \addplot coordinates { (0.809155,0.550622) };
    \addlegendentry{\texttt{r-rlz}};
    \addplot coordinates { (0.782948,0.11194) };
    \addlegendentry{\texttt{r-lzend}};
    \addplot coordinates { (50.0207,0.386802) };
    \addlegendentry{\texttt{move-r-rlz}};
    \addplot coordinates { (28.3184,0.141674) };
    \addlegendentry{\texttt{move-r-lzend}};
    \legend{}
    
    \nextgroupplot[
        title=\textbf{locate} ($m \ll \overline{\occ{}}$),
        yticklabels={},
        ymin=0,
        ymax=0.6,
        ytick distance=0.1,
    ]
    \addplot coordinates { (0.0590422,0.0596715) };
    \addlegendentry{\texttt{r-index}};
    \addplot coordinates { (1.15052,0.113669) };
    \addlegendentry{\texttt{move-r}};
    \addplot coordinates { (3.20213,0.550622) };
    \addlegendentry{\texttt{r-rlz}};
    \addplot coordinates { (0.711244,0.11194) };
    \addlegendentry{\texttt{r-lzend}};
    \addplot coordinates { (5.56663,0.386802) };
    \addlegendentry{\texttt{move-r-rlz}};
    \addplot coordinates { (0.646407,0.141674) };
    \addlegendentry{\texttt{move-r-lzend}};
    \legend{}
    
    
    \nextgroupplot[
        ylabel style={align=center},
        ylabel=\textbf{sars}\\{index size [bits/$n$]},
        ymin=0,
        ymax=1.5,
        ytick distance=0.25,
    ]
    \addplot coordinates { (9.78815,0.17418) };
    \addlegendentry{\texttt{r-index}};
    \addplot coordinates { (104.252,0.402362) };
    \addlegendentry{\texttt{move-r}};
    \legend{}
    
    \nextgroupplot[
        yticklabels={},
        ymin=0,
        ymax=1.5,
        ytick distance=0.25,
    ]
    \addplot coordinates { (0.156849,0.17418) };
    \addlegendentry{\texttt{r-index}};
    \addplot coordinates { (1.47018,0.402362) };
    \addlegendentry{\texttt{move-r}};
    \addplot coordinates { (0.221417,1.41642) };
    \addlegendentry{\texttt{r-rlz}};
    \addplot coordinates { (0.215509,0.400335) };
    \addlegendentry{\texttt{r-lzend}};
    \addplot coordinates { (2.88017,1.45354) };
    \addlegendentry{\texttt{move-r-rlz}};
    \addplot coordinates { (2.38155,0.553979) };
    \addlegendentry{\texttt{move-r-lzend}};
    \legend{}
    
    \nextgroupplot[
        yticklabels={},
        ymin=0,
        ymax=1.5,
        ytick distance=0.25,
    ]
    \addplot coordinates { (0.0116072,0.17418) };
    \addlegendentry{\texttt{r-index}};
    \addplot coordinates { (0.0458268,0.402362) };
    \addlegendentry{\texttt{move-r}};
    \addplot coordinates { (0.781238,1.41642) };
    \addlegendentry{\texttt{r-rlz}};
    \addplot coordinates { (0.183926,0.400335) };
    \addlegendentry{\texttt{r-lzend}};
    \addplot coordinates { (1.38618,1.45354) };
    \addlegendentry{\texttt{move-r-rlz}};
    \addplot coordinates { (0.187668,0.553979) };
    \addlegendentry{\texttt{move-r-lzend}};
    \legend{}
    
    
    \nextgroupplot[
        ylabel style={align=center},
        ylabel=\textbf{dewiki}\\{index size [bits/$n$]},
        ymin=0,
        ymax=1.25,
        ytick distance=0.25,
    ]
    \addplot coordinates { (0.0509048,0.26024) };
    \addlegendentry{\texttt{r-index}};
    \addplot coordinates { (0.85744,0.592061) };
    \addlegendentry{\texttt{move-r}};
    \legend{}
    
    \nextgroupplot[
        yticklabels={},
        ymin=0,
        ymax=1.25,
        ytick distance=0.25,
    ]
    \addplot coordinates { (1.08492,0.26024) };
    \addlegendentry{\texttt{r-index}};
    \addplot coordinates { (18.5936,0.592061) };
    \addlegendentry{\texttt{move-r}};
    \addplot coordinates { (1.23729,1.1723) };
    \addlegendentry{\texttt{r-rlz}};
    \addplot coordinates { (1.25095,0.484339) };
    \addlegendentry{\texttt{r-lzend}};
    \addplot coordinates { (26.2593,1.23814) };
    \addlegendentry{\texttt{move-r-rlz}};
    \addplot coordinates { (20.2994,0.696148) };
    \addlegendentry{\texttt{move-r-lzend}};
    \legend{}
    
    \nextgroupplot[
        yticklabels={},
        ymin=0,
        ymax=1.25,
        ytick distance=0.25,
    ]
    \addplot coordinates { (0.0453978,0.26024) };
    \addlegendentry{\texttt{r-index}};
    \addplot coordinates { (0.298952,0.592061) };
    \addlegendentry{\texttt{move-r}};
    \addplot coordinates { (1.9083,1.1723) };
    \addlegendentry{\texttt{r-rlz}};
    \addplot coordinates { (0.500247,0.484339) };
    \addlegendentry{\texttt{r-lzend}};
    \addplot coordinates { (4.00048,1.23814) };
    \addlegendentry{\texttt{move-r-rlz}};
    \addplot coordinates { (0.543672,0.696148) };
    \addlegendentry{\texttt{move-r-lzend}};
    \legend{}

    
    \nextgroupplot[
        ylabel style={align=center},
        ylabel=\textbf{chr19}\\{index size [bits/$n$]},
        ymin=0,
        ymax=11,
        ytick distance=2,
    ]
    \addplot coordinates { (0.0185479,1.91293) };
    \addlegendentry{\texttt{r-index}};
    \addplot coordinates { (0.250535,4.92133) };
    \addlegendentry{\texttt{move-r}};
    \legend{}
    
    \nextgroupplot[
        yticklabels={},
        ymin=0,
        ymax=11,
        ytick distance=2,
    ]
    \addplot coordinates { (0.0145336,1.91293) };
    \addlegendentry{\texttt{r-index}};
    \addplot coordinates { (0.240741,4.92133) };
    \addlegendentry{\texttt{move-r}};
    \addplot coordinates { (0.0139259,10.2501) };
    \addlegendentry{\texttt{r-rlz}};
    \addplot coordinates { (0.0155166,4.09237) };
    \addlegendentry{\texttt{r-lzend}};
    \addplot coordinates { (0.244519,9.24179) };
    \addlegendentry{\texttt{move-r-rlz}};
    \addplot coordinates { (0.233718,5.9234) };
    \addlegendentry{\texttt{move-r-lzend}};
    \legend{}
    
    \nextgroupplot[
        yticklabels={},
        ymin=0,
        ymax=11,
        ytick distance=2,
    ]
    \addplot coordinates { (0.0357893,1.91293) };
    \addlegendentry{\texttt{r-index}};
    \addplot coordinates { (0.55541,4.92133) };
    \addlegendentry{\texttt{move-r}};
    \addplot coordinates { (1.05697,10.2501) };
    \addlegendentry{\texttt{r-rlz}};
    \addplot coordinates { (0.324485,4.09237) };
    \addlegendentry{\texttt{r-lzend}};
    \addplot coordinates { (3.09591,9.24179) };
    \addlegendentry{\texttt{move-r-rlz}};
    \addplot coordinates { (0.466078,5.9234) };
    \addlegendentry{\texttt{move-r-lzend}};
    \legend{}
    
    
    \nextgroupplot[
        ylabel style={align=center},
        ylabel=\textbf{english}\\{index size [bits/$n$]},
        ymin=0,
        ymax=70,
        ytick distance=10,
    ]
    \addplot coordinates { (7.93718,24.9084) };
    \addlegendentry{\texttt{r-index}};
    \addplot coordinates { (111.59,56.2821) };
    \addlegendentry{\texttt{move-r}};
    \legend{}
    
    \nextgroupplot[
        xlabel={query throughput [1/ms]},
        yticklabels={},
        ymin=0,
        ymax=70,
        ytick distance=10,
    ]
    \addplot coordinates { (5.16089,24.9084) };
    \addlegendentry{\texttt{r-index}};
    \addplot coordinates { (86.5565,56.2821) };
    \addlegendentry{\texttt{move-r}};
    \addplot coordinates { (4.73906,46.8262) };
    \addlegendentry{\texttt{r-rlz}};
    \addplot coordinates { (5.39799,45.1645) };
    \addlegendentry{\texttt{r-lzend}};
    \addplot coordinates { (94.5555,53.5499) };
    \addlegendentry{\texttt{move-r-rlz}};
    \addplot coordinates { (80.7361,65.0831) };
    \addlegendentry{\texttt{move-r-lzend}};
    
    \nextgroupplot[
        yticklabels={},
        ymin=0,
        ymax=70,
        ytick distance=10,
    ]
    \addplot coordinates { (0.0121844,24.9084) };
    \addlegendentry{\texttt{r-index}};
    \addplot coordinates { (0.0707793,56.2821) };
    \addlegendentry{\texttt{move-r}};
    \addplot coordinates { (0.818028,46.8262) };
    \addlegendentry{\texttt{r-rlz}};
    \addplot coordinates { (0.174984,45.1645) };
    \addlegendentry{\texttt{r-lzend}};
    \addplot coordinates { (2.13067,53.5499) };
    \addlegendentry{\texttt{move-r-rlz}};
    \addplot coordinates { (0.180324,65.0831) };
    \addlegendentry{\texttt{move-r-lzend}};
    \legend{}
    
    \end{groupplot}
\end{tikzpicture}
\end{center}
\caption{
    Size (in bits per input character) versus locate query throughput (queries per millisecond) of our implemented index data structures given medium ($m \approx \overline{\occ{}}$) or short ($m \ll \overline{\occ{}}$)patterns for the given inputs.
	For reference, we also give the count query throughput of the base index data structures, \texttt{r-index} and \texttt{move-r} for medium patterns.
}
\label{fig:results-locate}
\end{figure}

\subsection{Construction Performance}

We first look at the construction of the competing indexes.
Figure~\ref{fig:results-construction} shows the construction throughput as well as the peak memory usage during construction.

To no surprise, compressing the differential suffix array dominates the time and space needed for construction (comparing \texttt{r-index} and \texttt{move-r} to the variants storing a compressed suffix array).
Regarding the two different compression schemes, we see that LZ-End (\texttt{move-r-lzend} and \texttt{r-lzend}) is relatively slow to compute overall, but competitive with \texttt{r-rlz} regarding both time and space.

Our improved RLZSA construction from Section~\ref{sect:rlzsa}~(\texttt{move-r-rlz}), however, clearly outperforms the other variants that compress the suffix array:
it is faster by a factor of up to ten (einstein) and the required space is sometimes even lower than for just computing the $r$-index.
It also clearly outperforms the construction of our reimplementation of RLZSA~(\texttt{r-rlz}).

\subsection{Locate Query Performance}

We now look at locate queries for the two pattern sets described above (one query per pattern).
Figure~\ref{fig:results-locate} shows the query throughput as well as the size of the considered indexes.
For reference, we also give the throughput of count queries, which does not involve any compressed suffix arrays (because we only report the size of the corresponding suffix array interval, not its contents).

We can assert that the performance of \texttt{move-r} is somewhat improved over~\cite{DBLP:conf/wea/Bertram0N24} (the experiments there were done on the same machine). 
The trade-off compared to \texttt{r-index} remains the same: we require roughly twice the amount of space, but queries are considerably faster overall.

As expected, enhancing the $r$-index by compressed suffix arrays (\texttt{r-rlz} and \texttt{r-lzend}) considerably improves the performance of locate queries for patterns with many occurrences.
This confirms the results of~\cite{DBLP:conf/dcc/PuglisiZ21}.
We see how \texttt{r-rlz} achieves overall higher throughputs than \texttt{r-lzend} (by a factor of 4 for $m \ll \overline{\occ{}}$).
This is expected, as random access on RLZ-compressed data incurs only one cache miss per phrase, as opposed to up to $h$ cache misses for LZ-End.
However, we see that LZ-End achieves better compression, which is also confirmed in Table~\ref{tab:inputs} when comparing columns $\floor{n/\hat{z}_\text{R}}$ and $\floor{n/\hat{z}_\text{end}}$, making it a trade-off.

When enhancing Move-$r$ with compressed suffix arrays (\texttt{move-r-rlz} and \texttt{move-r-lzend}), the picture differs a bit.
Here, using LZ-End (\texttt{move-r-lzend}) can sometimes even slow down locate queries (e.g., on einstein and chr19).
Using RLZ (\texttt{rlzsa}), on the other hand, improves query performance by a great deal particularly for frequent patterns $m \ll \overline{\occ{}}$ (e.g., by a factor of over 16 for sars).
Again, however, LZ-End yields much better compression than RLZ in most cases (now comparing $\floor{n/\hat{z}_\text{R'}}$ and $\floor{n/\hat{z}_\text{end}}$ in Table~\ref{tab:inputs}).
Interestingly however, on english, the improved RLZSA construction (\texttt{move-r-rlz}) achieves better compression than LZ-End (\texttt{move-r-lzend}), which is a topic for further research.

Overall, our improved RLZSA (\texttt{move-r-rlz}) achieves better compression than that of~\cite{DBLP:conf/dcc/PuglisiZ21} (\texttt{r-rlz}).
This is particularly evident for einstein and chr19, where \texttt{move-r-rlz} is smaller than \texttt{r-rlz} despite storing more information (e.g., compare \texttt{move-r} against \texttt{r-index}).

\section{Conclusions and Future Work}
\label{sect:final}

We enhanced the recent $r$-index as well as Move-$r$ by compressed suffix arrays with efficient random access to speed up locate queries.
For this, we explored two different compression schemes: Relative Lempel-Ziv and LZ-End.
The experiments show that the idea works, confirming and expanding upon the results of~\cite{DBLP:conf/dcc/PuglisiZ21}.
We can achieve different trade-offs regarding construction performance, index size and query performance by choosing different combinations of index and compressed suffix arrays.
For both compression schemes, we gave new strategies and algorithms that improve upon their predecessors.

In future research, enhancing the subsampled $r$-index by Cobas et al.~\cite{DBLP:conf/cpm/CobasGN21} may be considered.
We also saw that reference construction for Relative Lempel-Ziv is still an interesting topic of research beyond~\cite{DBLP:conf/spire/GagiePV16,DBLP:conf/dcc/PuglisiZ21}.
By improving upon the segment selection strategy of~\cite{DBLP:conf/dcc/PuglisiZ21}, we were able to improve the quality of the reference and thus compression.

\clearpage
\bibliographystyle{plainurl}
\bibliography{literature}

\clearpage
\appendix
\section{Computing the LZ-End Parsing}
\label{appendix:lzend-s-ext}

We look at computing the LZ-End parsing for a given text $T \in \Sigma^*$ in practice.
The algorithm given originally by Kreft and Navarro~\cite{DBLP:conf/dcc/KreftN10} runs in time $\Oh{nh(\lg|\Sigma| + \lg\lg n)}$ using FM-index machinery.
Kempa and Kosolobov~\cite{DBLP:conf/esa/KempaK17} greatly improve this and state an algorithm that runs in time $\Oh{n}$, but we conjecture it to be hardly practical.

We focus on their surprisingly simple and practically competitive (albeit suboptimal) $\Oh{n \lg\lg n}$-time algorithm, which they implemented in the \emph{lz-end-toolkit}\footnote{lz-end-toolkit: \url{https://github.com/dominikkempa/lz-end-toolkit}} accompanying~\cite{DBLP:conf/dcc/KempaK17}.
However, we find that the description of this algorithm comes either too short (in~\cite{DBLP:conf/esa/KempaK17}) or gets somewhat lost in the details of surrounding work (in~\cite{DBLP:conf/dcc/KempaK17}).
Therefore, we choose to give a concise but comprehensible description here and apply a few further practical improvements.
The following observation is the main tool for their algorithm.
\begin{lemma}
\label{lemma:kk17}
If $f_1 \cdots f_z$ is the LZ-End parsing of a string $T \in \Sigma^*$, then, for any character $\alpha \in \Sigma$,
the last phrase in the LZ-End parsing of $T\alpha$ is (1)~$f_{z-1} f_z \alpha$ or (2)~$f_z \alpha$ or (3) $\alpha$.
\end{lemma}

This allows us to compute the LZ-End parsing in a left-to-right scan of $T$ where in each step,
we only have to consider to either
(1)~\emph{merge} the two most recent phrases,
(2)~\emph{extend} the most recent phrase or
(3)~\emph{begin} a new phrase consisting of a single character.

Suppose that we already parsed the prefix $T[1 ~..~ i-1] = f_1 \cdots f_z$ for some position $i > 1$.
In the next step, we compute the LZ-End parsing of $T[1 \dots i]$, i.e., we append $T[i]$.
We look for a phrase $f_p$ such that a suffix $X$ of $T[1 ~..~ i-1]$ of maximum possible length is also a suffix of $T[1 ~..~ |f_1 \cdots f_p|]$.
Then, we greedily decide which of the three aforementioned cases applies.
If $|X| \geq |f_{z-1}|+|f_z|$, it means that the new phrase covers at least the two most recent phrases and we can merge them to $(p, |f_{z-1}|+|f_z|+1, T[i])$.
Otherwise, if $|X| \geq |f_z|$, we can extend the most recent phrase to $(p, |f_z|+1, T[i])$.
If neither applies, we begin a new phrase $(0,1,T[i])$.
The core of the problem is clearly finding $f_p$ and the corresponding suffix $X$.
We employ the following data structures over the \emph{reverse} input $\overleftarrow{T}$:
\begin{enumerate}
\item the inverse suffix array $A^{-1}$, which we can compute in linear time and space using \cite{DBLP:journals/tc/NongZC11} and subsequent inversion,
\item the LCP array $H$ that can be computed in linear time and space~\cite{DBLP:conf/cpm/KasaiLAAP01}, where $H[i] = \text{lce}(T[A[i-1] ~..~ n], T[A[i] ~..~ n)$ for $i > 1$ is the length of the longest common extension (lce) between two lexicographically neighbouring suffixes,
\item a data structure that allows for constant-time \emph{range minimum queries} (rmq) over $H$, which we can compute in linear time and space~\cite{DBLP:journals/siamcomp/FischerH11} and
\item an associative dynamic predecessor/successor data structure $M$ that is initially empty.
\end{enumerate}

In $M$, we mark the end positions of already computed LZ-End phrases in lex-space, i.e., a position $i \geq 1$ is marked iff there is a phrase $f_p$ such that $A[i] = |f_1 \cdots f_p|$.
We then use $M$ to lookup the phrase number $p$ using position $i$ as the query key.
We note that $A$ is used only conceptually and is only required for the construction of $A^{-1}$ and $H$; it may be discarded afterwards to save space.

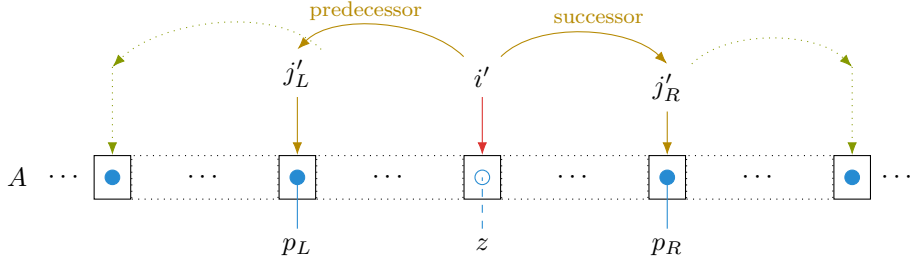
\begin{figure}[t]
\begin{center}
\begin{tikzpicture}[
    every node/.style={
        font=\small,
    },
    box/.style={
        draw,
        shape=rectangle,
        minimum height=1.5em,
        minimum width=1.25em
    },
    dots/.style={
        box,
        dotted,
        minimum width=5em,
    },
    virt/.style={
        minimum height=1.5em,
    }
]
\node(llldots){$\cdots$};
\node[left=0 of llldots]{$A$};
\node[box,right=0 of llldots](pl2){~};
\node[dots,right=0 of pl2](lldots){$\cdots$};
\node[box,right=0 of lldots](pl1){~};
\node[below=1em of pl1](pl1-label){$p_L$};
\node[dots,right=0 of pl1](ldots){$\cdots$};
\node[box,right=0 of ldots](z){~};
\node[below=1em of z](z-label){$z$};
\node[dots,right=0 of z](rdots){$\cdots$};
\node[box,right=0 of rdots](pr1){~};
\node[below=1em of pr1](pr1-label){$p_R$};
\node[dots,right=0 of pr1](rrdots){$\cdots$};
\node[box,right=0 of rrdots](pr2){~};
\node[right=0 of pr2]{$\cdots$};

\draw[solarized-blue,fill=solarized-blue] (pl2) circle (.25em);
\draw[solarized-blue,fill=solarized-blue] (pl1) circle (.25em);
\draw[solarized-blue] (pl1-label) -- (pl1.center);
\draw[solarized-blue] (z) circle (.25em);
\draw[solarized-blue,dashed] (z-label) -- (z.center);
\draw[solarized-blue,fill=solarized-blue] (pr1) circle (.25em);
\draw[solarized-blue] (pr1-label) -- (pr1.center);
\draw[solarized-blue,fill=solarized-blue] (pr2) circle (.25em);

\node[above=2em of z](i){$i'$};
\draw[solarized-red,-Latex] (i) -- (z);

\node[above=2em of pl1](jl1){$j'_L$};
\draw[solarized-yellow,-Latex] (jl1) -- (pl1);
\draw[solarized-yellow,-Latex] (i.north west) to[out=135,in=45,looseness=0.8] node[yshift=0.5em,xshift=-0.5em]{\scriptsize{predecessor}} (jl1.north);

\node[virt,above=1.5em of pl2](jl2){~};
\draw[solarized-green,dotted,-Latex] (jl1.north east) to[out=135,in=45,looseness=0.8] (jl2.north);
\draw[solarized-green,dotted,-Latex] (jl2.north) -- (pl2);

\node[above=1.5em of pr1](jr1){$j'_R$};
\draw[solarized-yellow,-Latex] (jr1) -- (pr1);
\draw[solarized-yellow,-Latex] (i.north east) to[out=45,in=135,looseness=0.8] node[yshift=0.5em,xshift=0.5em]{\scriptsize{successor}} (jr1.north);

\node[virt,above=1.5em of pr2](jr2){~};
\draw[solarized-green,dotted,-Latex] (jr1.north east) to[out=45,in=135,looseness=0.8] (jr2.north);
\draw[solarized-green,dotted,-Latex] (jr2.north) -- (pr2);
\end{tikzpicture}
\end{center}
\caption{
    Search for source phrase candidates in lex-space.
    Positions that mark the ending locations in $M$ of already computed LZ-End phrases are indicated by the circles.
    We search for a predecessor starting from $i'-1$ and a successor starting from $i'+1$, giving us the locations $j'_L$ and $j'_R$ that mark the candidates $p_L$ and $p_R$, respectively.
    In case $p_L = z-1$ or $p_R = z-1$, we do another predecessor or successor query starting from $j'_L-1$ or $j'_R+1$, respectively, to find a candidate for merging.
}
\label{fig:lzend-candidates}
\end{figure}

In the situation described earlier, we already parsed $T[1 ~..~ i-1] = f_1 \cdots f_z$ and want to find the longest possible suffix $X$ that is a suffix of $T[1 ~..~ |f_1 \cdots f_p|]$ for some phrase $f_p$.
In $M$, we search for the predecessor $p_L$ and the successor $p_R$, respectively, starting from position $i' = A^{-1}[n-i]$.
The phrases $f_{p_L}$ and $f_{p_R}$ are candidates for our sought phrase $f_p$ because the strings $\overleftarrow{T[1 ~..~ |f_1 \cdots f_{p_L}|]}$ and $\overleftarrow{T[1 ~..~ |f_1 \cdots f_{p_R}|]}$ are lexicographically closest to $\overleftarrow{T[1 ~..~ i-1]}$.
We greedily select $p := p_L$ or $p := p_R$ by whichever shares a longer common prefix, which we can compute via a range minimum query over $H$, respectively.

Some care has to be taken regarding the selection of $p$: we need $p \neq z$ (we cannot copy from the last phrase because we extend it) and for merging, we also need $p \neq z-1$ (we cannot copy from a phrase that we merge away).
We ensure $p \neq z$ (which is marked at position $i'$ in $M$ because it is $|f_1 \cdots f_z| = i$) by offsetting predecessor searches to start from position $i'-1$ and successor searches to start from position $i'+1$.
To ensure $p \neq z-1$ for merges, we have to check whether $p_L = z-1$ or $p_R = z-1$ and if either is the case, compute the next predecessor or successor, respectively, conceptually skipping phrase $f_{z-1}$ in $M$.

The candidate search is visualized in Figure~\ref{fig:lzend-candidates}.
To get overall time $\Oh{n \lg\lg n}$ to compute the LZ-End parsing, we can implement $M$ using a y-fast trie~\cite{DBLP:journals/ipl/Willard83} such that each update and query can be done in time $\Oh{\lg\lg n}$.
It is easy to see that the required data structures require space $\Oh{n}$.
For further reference, we give the pseudocode for the parsing algorithm in Algorithm~\ref{algo:lzend}, which is deliberately close to the actual C++ implementation.

\begin{table}[t]
\begin{center}
\begin{tabular}{r | r r r | r r}
Input & $n$ & $z$ & $z/n$ & \emph{lz-end-toolkit} & Algorithm~\ref{algo:lzend}\\
\hline
cere            & \num{461286644}  & \num{1863246}  & \num{0.40}\%  & \underline{\num{455}} & \num{582} \\
dblp.xml        & \num{296135874}  & \num{10244979} & \num{3.46}\%  & \num{283} & \underline{\num{251}} \\
dna             & \num{403927746}  & \num{26939573} & \num{6.67}\%  & \num{602} & \underline{\num{379}} \\
einstein.en.txt & \num{467626544}  & \num{104087}   & \num{0.02}\%  & \underline{\num{454}} & \num{728} \\
english.1024MB  & \num{1073741824} & \num{68034586} & \num{6.34}\%  & \num{2119} & \underline{\num{1160}} \\
pitches         & \num{55832855}   & \num{5675142}  & \num{10.16}\% & \num{27} & \underline{\num{25}} \\
proteins        & \num{1184051855} & \num{77369007} & \num{6.53}\%  & \num{1651} & \underline{1185} \\
sources         & \num{210866607}  & \num{12750341} & \num{6.05}\%  & \num{158} & \underline{135} \\
\end{tabular}
\end{center}
\caption{
    Benchmark results showing the parsing times, in seconds, of \emph{lz-end-toolkit} and Algorithm~\ref{algo:lzend} (shortest underlined).
    For each input, we also list the length $n$, the number $z$ of LZ-End phrases and the ratio $z/n$ as a simple compressibility measure.
}
\label{tab:lzend-bench}
\end{table}

\subsection{Experiments}

In a small experiment on the same setup as that described in Section~\ref{sect:experiments}, we compress different files from the Pizza\&Chili corpus and compare the performance of our implementation
\footnote{
    We use \emph{libsais} for computing the suffix and LCP arrays, the data structure of \cite{DBLP:conf/latin/BenderF00} for range minimum queries, and a B-tree of \cite{DBLP:conf/wea/Dinklage0H21} as an ordered dictionary to implement $M$.
}
of Algorithm~\ref{algo:lzend} against the in-memory implementation featured in the \emph{lz-end-toolkit} by~\cite{DBLP:conf/dcc/KempaK17}.
Here, the maximum phrase length remains unbounded ($h := \infty$).
The results are given in Table~\ref{tab:lzend-bench}.

On most inputs, our implementation is much faster (up to nearly twice as fast on \emph{english.1024MB}) than the \emph{lz-end-toolkit}, indicating that the lazy evaluation of predecessor and successor queries as well as the removed layer of indirection via the suffix array $A$ can be very beneficial.

It stands out, however, that the \emph{lz-end-toolkit} is faster than Algorithm~\ref{algo:lzend} on highly repetitive inputs (namely \emph{cere} and \emph{einstein.en.txt}).
The reason is that it is implicitly tuned for this case:
the temporary removal of phrase $f_{z-1}$ from $M$ is beneficial for the case that merges occur frequently.

We can easily tune Algorithm~\ref{algo:lzend} for this case as well by re-arranging the candidate search to first look for merges, including a preliminary removal of $f_{z-1}$ from $M$ to take advantage of.
However, pointing at the results given in Table~\ref{tab:lzend-bench}, we conjecture that this slows down the algorithm in the general case.

\begin{algorithm}[p]
\DontPrintSemicolon
\SetKwFunction{FParse}{\textsc{LZ-End}}
\SetKwFunction{FLeftPhrase}{\textsc{LexSmallerPhrase}}
\SetKwFunction{FRightPhrase}{\textsc{LexGreaterPhrase}}
\SetKwFunction{FFindCandidate}{\textsc{FindCopySource}}
\begin{small}
\Fn{\FParse{$T$}}{
    \tcc{build index and initialize}
    $A \gets $ suffix array of $\overleftarrow{T}$,
    $H \gets $ LCP array of $\overleftarrow{T}$\ (with rmq support)\;
    $A' \gets $ array of length $n$\;
    \lFor{$i \gets 0$ \KwTo $n-1$}{
        $A'[n-A[i]-1] = i$
    }
    discard $A$ and $\overleftarrow{T}$\;
    $f_0 \gets (0,~0,~\varepsilon)$,
    $f_1 \gets (0,~1,~T[0])$,
    $z \gets 1$,
    $M \gets \emptyset$\;
    \For{$i \gets 1$ \KwTo $n-1$}{\label{algo:lzend:main-loop}
        $i' \gets A'[i-1]$ \tcp{suffix array neighbourhood of $i-1$ in $\overleftarrow{T}$}
        $p_1 \gets \bot$,
        $p_2 \gets \bot$\;
        \tcc{find candidates}
        \FFindCandidate{\FLeftPhrase}\;\label{algo:lzend:left}
        \If{$p_1 = \bot$ \Or $p_2 = \bot$}{
            \FFindCandidate{\FRightPhrase}\;\label{algo:lzend:right}
        }
        \tcc{case distinction according to Lemma~\ref{lemma:kk17}}
        \If{$p_2 \neq \bot$}{
            \tcc{merge phrases $f_{z-1}$ and $f_z$}
            $M \gets M \setminus \{ (A'[i-|f_z|-1], \ast) \}$ \tcp{unmark phrase $f_{z-1}$}\label{algo:lzend:unmark}
            $f_{z-1} \gets (p_2, ~|f_z| + |f_{z-1}| + 1, ~T[i])$\;
            $z \gets z-1$\;
        }\ElseIf{$p_1 \neq \bot$}{
            \tcc{extend phrase $f_z$}
            $f_{z} \gets (p_1, ~|f_z| + 1, ~T[i])$\;
        }\Else{\label{algo:lzend:new}
            \tcc{begin new phrase}
            $M \gets M \cup \{ i', z \}$ \tcp{lazily mark phrase $f_z$}
            $f_{z+1} \gets (0, ~1, ~T[i])$\;
            $z \gets z+1$\;
        }
    }
    \Return $(f_1, \dots, f_z)$\;
}
\;
\Fn{\FFindCandidate{$f$}}{
    $(j', p, \ell) \gets f(i')$\;
    \If{$\ell \geq |f_z|$}{
        $p_1 \gets p$\;\label{algo:lzend:extend}
        \If{$i > |f_z|$}{
            \lIf{$p = z-1$}{
                $(j', p, \ell_L) \gets f(j')$
            }\label{algo:lzend:left-phrase2}
            \lIf{$\ell \geq |f_z| + |f_{z-1}|$}{
                $p_2 \gets p$\label{algo:lzend:merge}
            }
        }
    }
   \Return $(p_1, p_2)$\;
}
\Fn{\FLeftPhrase{$i'$}}{
    \If{$M$ contains the predecessor $j' \leq i'-1$ with $j' \mapsto p$}{
        \Return $(j',p, H[\text{rmq}(j'+1, i')])$\;
    }\lElse{
        \Return $(0,0,0)$
    }
}
\Fn{\FRightPhrase{$i'$}}{
    \If{$M$ contains the successor $j' \geq i'+1$ with $j' \mapsto p$}{
        \Return $(j',p, H[\text{rmq}(i'+1, j')])$\;
    }\lElse{
        \Return $(0,0,0)$
    }
}
\end{small}
\vspace{1ex}
\caption{
    Algorithm to compute the LZ-End parsing for a string $T \in \Sigma^n$.
    Note that here, we use zero-based indexing to more closely resemble the implementation.
    For $i \in [0,n-1]$ and $p \in [0,z]$, we say that $i \mapsto p$ iff position $i$ is marked in $M$ and phrase $f_p$ ends at the corresponding location.
}
\label{algo:lzend}
\end{algorithm}

\section{Engineering RLZSA}
\label{appendix:rlzsa}
Computing the RLZ-compressed suffix array (RLZSA) according to~\cite{DBLP:conf/dcc/PuglisiZ21} requires time and space $\Oh{n}$.
More precisely, it requires five integer arrays of length $n$ in RAM.
This results in a much higher memory consumption than that of the $r$-index or Move-$r$, which need only $\Oh{|\mathsf{PFP}|}$ space, where $\mathsf{PFP}$ is a prefix-free parsing of the input $T$ with $|\mathsf{PFP}| \ll n$.
This motivates our optimized construction algorithm, which requires only $\Oh{r}$ space and $\Oh{r^{1 - \epsilon}n^\epsilon}$ time for a parameter $\epsilon \in [0, 1]$ (see Appendix~\ref{appendix:rlzsa_ref_constr}), making RLZSA practical for indexing large texts.

In order to reason about our optimized implementation, we first briefly summarize the original implementation from \cite{DBLP:conf/dcc/PuglisiZ21} in Appendix~\ref{appendix:rlzsa_original}.
Then, we discuss our adjusted set of RLZSA index data structures, our new reference construction algorithm and practical optimizations for the RLZ parsing algorithm.

\begin{definition}
	\label{def:rlz}
	Given a reference $R \in \Sigma^*$, we call $\langle s_1, l_1 \rangle, \langle s_2, l_2 \rangle, \dots, \langle s_z, l_z \rangle$ a \emph{Relative Lempel-Ziv (RLZ) parsing of $T$ w.r.t.\ $R$} if $\sum_{i = 1}^z \max(1, l_i) = n$ and $T[p_i, p_i + l_i) = R[s_i, s_i + l_i)$ is the longest prefix of $T_{p_i}$ that occurs in $R$, if it exists, or $\langle s_i, l_i \rangle = \langle T[p_i], 0 \rangle$, else, for each $i \in [1, z]$ and $p_i = 1 + \sum_{j = 1}^{i - 1} \max(1, l_j)$.
\end{definition}
\begin{definition}
	Let $S$ and $x$ be sequences and let $k \geq 1$ be an integer.
    Then, $\kmerset_k(S)$ denotes the set of $k$-mers that occur in $S$ and $\#S(x)$ denotes the number of occurrences of $x$ in $S$.
\end{definition}

\subsection{Summary of the original RLZSA Implementation}
\label{appendix:rlzsa_original}
The RLZSA index described in \cite{DBLP:conf/dcc/PuglisiZ21} consists of the $r$-index -- without a data structure for computing $\Phi$ and without suffix array samples -- and the RLZSA, represented using the following arrays:

\begin{itemize}
	\item $\R$ -- the reference $R$ stored in $\Oh{|R| \lg n}$ bits,
	\item $\Sarr[1~..~z] = [s_1, \dots, s_z]$ -- the phrase source positions in $R$ or literal phrases, respectively, stored in $\Oh{z \lg n}$ bits,
	\item $\PL[1~..~z] = [l_1, \dots, l_z]$ -- the phrase lengths, stored also in $\Oh{z \lg n}$ bits and
	\item $\PS[1~..~\lfloor z/a \rfloor]$ -- samples of phrase starting positions where $\PS[i] = p_{ai}$, stored in $\Oh{z/a \lg n}$ bits.
\end{itemize}

Here, $a \geq 1$ is an integer sampling parameter.
We set $a := 64$ as described by the authors.
To reduce space usage in practice, the copy phrase length has been limited to $2^{16}$.
$R$ is constructed by dividing $\Ad$ into consecutive segments $S_1,\dots,S_{n'}$ of size $s$, scoring those, and iteratively adding a segment that maximizes the score until $R$ has reached its target size.
A segment's score rises with the frequencies in $A^d$ of the new $k$-mers that it adds to $R$.
More formally, the score of segment $S_i = \Ad[is~..~(i + 1)s)$ is
\begin{equation*}
\end{equation*}
where $C \subseteq [1, n']$ is the set containing the indices of already chosen segments.

The reference sizes in \cite{DBLP:conf/dcc/PuglisiZ21} have been tuned manually for each input.
Regarding the RLZ parsing, it is modified such that every referencing phrase is preceded by a literal phrase.
This simplifies the query procedure and reduces query time in the case that the start of the interval to extract lies at the end of a long series of long copy phrases. 

\subsection{Index Data Structures}
\label{appendix:rlzsa_data_structures}

We begin by discussing our choice of index data structures.
It consits of the arrays

\begin{itemize}
	\item $\R$ -- the reference $R$ stored in $\Oh{|R| \lg n}$ bits,
	\item $\PT[1~..~z]$ -- phrase types, where $\PT[i] := 1 \Leftrightarrow l_i = 0$ stored as a bit vector with constant-time rank/select support in $z+\oh{z}$ bits,
	\item $\LP[1~..~z_l]$ -- literal phrases, where $\LP[i] := s_j$ with $j = \PT.\select_1(i)$, stored in $\Oh{z_l \lg n}$ bits (where $z_l$ is the number of literal phrases),
	\item $\SR[1~..~z_c]$ -- phrase source positions in $R$, where $\SR[i] := s_j$ with $j = \PT.\select_0(i)$ stored in $\Oh{z_c \lg |R|}$ bits (where $z_c$ is the number of copy phrases),
	\item $\CPL[1~..~z_c]$ -- copy phrase lengths, where $\CPL[i] := l_j$ with $j = \PT.\select_0(i)$ stored in $\Oh{z_c \lg (n / z)}$ bits) and
	\item $\SCP[1~..~\lfloor z_c / a \rfloor]$ -- samples of copy phrase starting positions in $T$, where $\SCP[i] := p_j$ with $j = \PT.\select_0(ai)$, stored as an s-array~\cite{DBLP:conf/alenex/OkanoharaS07,sdsl} using $(z_c / a) \lg (n a /z_c) + 2 z_c / a + \oh{z_c / a}$ bits.
\end{itemize}

Here, $a \geq 1$ is an integer sampling parameter.
In practice, we use $a = 4$.
Since $|R| \ll n$ in practice, storing $\ceil{\lg n}$ bits per value in $\Sarr$ is wasteful.
Therefore, we split $\Sarr$ up into two arrays $\SR$ and $\LP$ and the bit vector $\PT$.
Then, we can store $\SR$ with $\ceil{\lg |R|}$ bits per value.
As in \cite{DBLP:conf/dcc/PuglisiZ21}, we limit the copy phrase length to $2^{16}$ such that we can store $\CPL$ using 16 bits per value.
Finally, we replace copy phrases of length one with literal phrases.
This reduces the number of cache misses caused by lookups in $R$ and reduces the index size, because one value in $\SR$ and $\CPL$, respectively, and $1/a$ values in $\SCP$ are replaced by one value in $\LP$.

Additionally, we store $\MLF$, $L'$ from Move-$r$~\cite{DBLP:conf/wea/Bertram0N24} and $\RSLap$ from Appendix~\ref{appendix:rank_select} to compute the suffix array interval of a pattern.
Finally, we store the array $\SAs[1..r']$, where $\SAs[i] = \A[\MLF.p[i]]$ and maintain $z(i)$ and $\hat{b}'_{z(i)}$ during the backward search phase of a locate query, where $z(i) = A[b_i]$ is defined analogously to~\cite[Definition~12]{DBLP:conf/wea/Bertram0N24}.
Then, we can compute $A[b] = \SAs[\hat{b}'_{z(1)}] - z(1)$ in constant time after the backward search.
This eliminates the need for the rule that each copy phrase has to be preceded by a literal phrase, as we do not have to decode the region between the suffix array interval and the last literal phrase before it.
Additionally, this reduces the overall number of phrases by a factor up to two.

\subsection{Queries}
\label{appendix:rlzsa_queries}

We now show how to answer queries using our data structures.
Storing $\SR$, $\LP$ and $\PT$ instead of $\Sarr$ makes the query procedure more complicated, because we now have to also maintain the indices $\xcp$ and $\xlp$ of the current copy- and literal phrases, respectively, i.e, $x$ is the index of the phrase containing $b + 1$, $\xcp$ is the index of the last copy-phrase starting at or before $b + 1$ and $\xlp$ is the index of the last literal phrase at or before $b + 1$.

Lastly, we need the starting position $p_x$ of the phrase containing $b + 1$.

\begin{figure}[t]
\begin{center}
\includegraphics[width=0.6\textwidth]{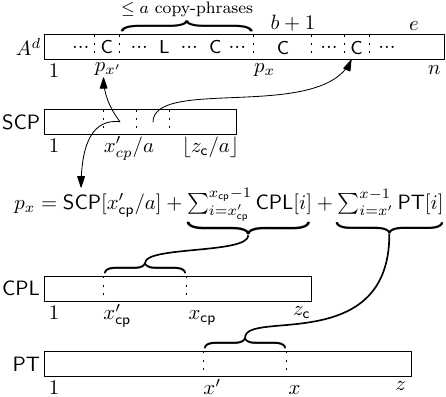}
\end{center}
\caption{
    Illustration of a the initialization phase of an RLZSA query. "\textsf{C}" and "\textsf{L}" indicate that the interval in $A^d$ is a copy/literal phrase.
}
\label{fig:rlzsa-query}
\end{figure}

To compute those values, we at first compute a lower bound $\xcp' \leftarrow a \cdot \SCP.\rank_1(b)$ for $\xcp$ (see Figure~\ref{fig:rlzsa-query}).
This takes time $\Oh{\lg(na / z_c)}$~\cite{sdsl}.
$x' \leftarrow \PT.\select_0(\xcp')$ gives us the phrase index of the $\xcp'$-th copy phrase.
Then, we compute its starting position $p_{x'} \leftarrow \SCP.\select_1(\xcp' / a)$ in constant time~\cite{sdsl}.
Using $x'$, $\xcp'$ and $p_{x'}$, we can then traverse RLZSA to the right until the $x'$-th phrase contains $b + 1$.
More formally, we compute $x$, $\xlp$ and $p_x$ by the equation in Figure~\ref{fig:rlzsa-query}.
Finally, $\xlp \leftarrow x - \xcp$ gives us the current literal phrase index.
This takes overall time $\Oh{a}$ if we use $\PT.\select_0$ queries to skip blocks of consecutive literal phrases in $\Oh{1}$ time.
Decoding $A(b, e]$ using $x$, $\xcp$, $\xlp$ and $p_x$ is similar and takes optimal time $\Oh{|[b, e]|}$.

\subsection{Index Construction}
\label{appendix:rlzsa_ref_constr}

In the following, we describe the construction of RLZSA.
After constructing $A$, $\MLF$, $L'$, $\RSLap$ and $\SAs$ in $\Oh{n + r \lg r}$ time, $\Oh{n}$ external space and $\Oh{|\mathsf{PFP}|}$ space in the RAM using Big-BWT~\cite{pfp}, we will only access $A$ in external memory.
Our algorithm for constructing $R$ follows a similar method as the algorithm presented in \cite{DBLP:conf/dcc/PuglisiZ21}, but reduces the running time from $\Oh{n}$ to $\Oh{r^{1 - \epsilon} n^\epsilon}$ for a parameter $\epsilon \in [0, 1]$.
The space is also reduced from $\Oh{n}$ to $\Oh{r}$.
Instead of splitting $\Ad$ into consecutive segments of size $s$, we consider arbitrary segments of size $s$ and set the $k$-mer length to $k = 1$, because when setting the segment size optimally ($s = 3072$), the number of phrases in RLZSA only rises when increasing $k$ beyond 1.
Choosing $k = 1$ also simplifies the computation of all $1$-mer frequencies.
We show how this can be done efficiently in the following section.
Then, we describe the construction of the reference.

\subsubsection{Computing the frequencies of all values in $\Ad$}

To see how we can compute the frequencies of all values ($1$-mers) in $\Ad$, we need the following lemma from~\cite{DBLP:conf/soda/GagieNP18}.
\begin{definition}
    Let $l_1, \dots, l_r$ be the starting positions of the runs in $L$, and let $l_{r + 1} = n + 1$.
    Let $\Phi$ be a function such that $\Phi(A[i]) = A[(i - 1) \mod n]$.
\end{definition}
\begin{lemma}[\cite{DBLP:conf/soda/GagieNP18}, Lemma 3.5]\label{lemma:phi}
    Let $\{u_1, u_2, \dots, u_{r+1}\} = \{\A[l_1], \A[l_2], \dots, \A[l_r], n+1\}$ and $u_1 < u_2 < \dots < u_{r+1} = n+1$. Then $\Phi(i) = \Phi(u_x) + (i-u_x)$ for $u_x \le i < u_{x+1}$.
\end{lemma}

We now show the following.
\begin{theorem}
	\label{thm:freq_sad}
	($i$) Given $\IPhi = (u_1, \Phi(u_1)), \dots, (u_r, \Phi(u_r))$, we can compute the frequencies of all values in $\Ad$ in $\Oh{r}$ expected time and space, and ($ii$) there are $\leq r + 1$ distinct values in $\Ad$.
\end{theorem}
\begin{proof}
    We compute a hash map $\HAd$ that maps $\langle \Ad[i] \rightarrow \#\Ad(\Ad[i]) \rangle$, for $i \in [1, n]$.
    We start by inserting $\langle A[1] \rightarrow 1 \rangle$ into $\HAd$.
    Then, we iterate with $x$ from $1$ to $r$.
    Each value $i \in [u_x, u_{x + 1})$ is mapped to $\Phi(i) = \Phi(u_x) + (i - u_x)$ by \cite[Lemma~3.5]{DBLP:conf/soda/GagieNP18}.
    Hence $\Ad[\ISA[i]] = i - \Phi(i) = u_x - \Phi(u_x)$ for $i \neq A[1]$.
    Let $v = u_x - \Phi(u_x)$ and $f = u_{x + 1} - u_x$, if $x \neq r$, and $f = u_{x + 1} - u_x - 1$, else (this avoids counting $A[1] - A[n]$).
    Now, we check, whether there is a mapping $\langle v \rightarrow f' \rangle \in \HAd$.
    If so, then we increment $f'$ by $f$.
    Else, we insert $\langle v \rightarrow f \rangle$ into $\HAd$.
    Since the intervals $[u_x, u_{x + 1})$ are disjoint, we consider $i = A[j]$ in exactly one iteration, for each $j \in [2, n]$.
    Hence, this algorithm correctly computes $\HAd$.
    
    Note that the number of mappings in $\HAd$ is equal to the number of distinct values in $\Ad$.
    This number is at most $\leq r + 1$, because we add at most one mapping for each one of the $r$ intervals $[u_x, u_{x + 1})$, and separately handling $\Ad[1] = A[1]$ adds at most one extra value.
    Thus, we have shown $(i)$ and $(ii)$.
\end{proof}

\subsubsection{Reference Construction}
\label{appendix:reference_construction}

Our algorithm uses $\HAd$ from Theorem~\ref{thm:freq_sad} to iteratively choose segments.
We maintain a balanced search tree $\Ts = \{\langle b_1, e_1 \rangle ,~ \dots,~ \langle  b_N, e_N \rangle\}$ with $1 < b_1 < e_1 < \dots < b_N < e_N < n$ to represent the current state of $R = \Ad[b_1 ~..~ e_1]\Ad[b_2 ~..~ e_2] \dots \Ad[b_N ~..~ e_N]$.
We also maintain that $\HAd$ maps $\langle \Ad[i] \rightarrow \#\Ad(\Ad[i]) \rangle$, if $\Ad[i] \notin R$, and $\langle \Ad[i] \rightarrow 0 \rangle$, else, for $i \in [1, n]$.
We set $t_R = \Oh{r}$ as the target size for $R$. 

\paragraph{Scoring segments.}
In each iteration, we consider $M = \Oh{(n/r)^\epsilon}$ (with $\epsilon \in [0, 1]$) random candidate segments $[l_1, r_1], \dots, [l_M, r_M] \subseteq [1, n]$ of fixed length $s$.
Given a candidate segment $[l_i, r_i]$, we at first shorten it from the left and/or the right (using a successor search over $\Ts$), such that it does not intersect the already chosen segments.
More precisely, we instead consider the segment $[l'_i, r'_i] = [l_i, r_i] \setminus \cup_{j \in [1, M]} [b_j, e_j]$.
This segment must be connected, because $|[b_j, e_j]| \geq s \, \forall j \in [1, M]$.
This takes time $\Oh{\lg N} = \Oh{\lg (t_R / s)}$.
If $[l'_i, r'_i] \neq \emptyset$, then we compute its score
\begin{align*}
	f([l'_i, r'_i]) &= \left(\sum_{x \in \kmerset_1(\Ad[l'_i ~..~ r'_i]) \setminus \kmerset_1(R)} \sqrt{\#\Ad(x)}\right) / |[l'_i, r'_i]| \\
    &= \left(\sum_{x \in \kmerset_1(\Ad[l'_i ~..~ r'_i])} \sqrt{\HAd[x]}\right) / |[l'_i, r'_i]|.
\end{align*}
This can be done in expected time $\Oh{|[l'_i ~..~ r'_i]|} = \Oh{s}$ by scanning over $\Ad[l'_i ~..~ r'_i]$ once and maintaining a temporary hashtable storing the already considered values in $\Ad[l'_i ~..~ r'_i]$.
Thus, scoring all segments takes $\Oh{M (s + \lg(t_R/s))}$ expected time.

\paragraph{Adding the best segment.}
Let $[l'_m, r'_m]$ be a segment that maximizes the score.
We update $\HAd$ to map $\langle \Ad[j] \rightarrow 0 \rangle$ for each $j \in [l'_m, r'_m]$ in $\Oh{s}$ expected time.
To reduce the number of segments stored in $\Ts$ and thereby also memory consumption, we merge $[l'_m, r'_m]$ with already chosen directly adjacent segments.
More precisely, if there exists a pair $\langle b_x, e_x \rangle \in \Ts$ with $e_x + 1 = l'_m$, then we remove $\langle b_x, e_x \rangle$ from $\Ts$ and set $l'_m \leftarrow b_x$.
Similarly, if there exists a pair $\langle b_y, e_y \rangle \in \Ts$ with $b_y - 1 = r'_m$, then we remove $\langle b_y, e_y \rangle $ from $\Ts$ and set $r'_m \leftarrow e_y$. Finally, we insert $\langle l'_m, r'_m \rangle$ into $\Ts$.
We stop as soon as $|R| \geq (1 - \epsilon') t_R$, where $\epsilon' \in [0, 1]$.
The search in $\Ts$ takes time $\Oh{\lg (t_R / s)}$ time.
Thus, adding one segment takes overall time $\Oh{s + \lg (t_R / s)}$.

\paragraph{Running time and memory consumption.}
If we assume that the expected length of the newly added segment $[l'_m, r'_m]$ (before merging) is $\Th{s}$, then this algorithm takes overall expected time $\Oh{(t_R / s) \cdot M \cdot (s + \lg (t_R / s))} = \Oh{r \cdot (n / r)^\epsilon} = \Oh{r^{1 - \epsilon} n^\epsilon}$ (for $\lg (t_R / s) = \Oh{s}$) and space $\Oh{r + t_R / s} = \Oh{r}$.
In practice, we set $t_R := \min(5.2 r, n / 3)$, $s := 3072$, $M := 5 (n / r)^\epsilon$, $\epsilon = 0.45$ and $\epsilon' = 1/20$.

\paragraph{Post processing.}
As a post-processing step, we close short gaps between long adjacent segments.
We maintain $\Ts$ and additionally a balanced search tree $\Tg$ initialized with $\Tg = \{\langle g_1, s_1 \rangle, \dots, \langle g_{N - 1}, s_{N - 1} \rangle\}$, where initially $g_i = b_i$ and $s_i = |[b_i, e_{i + 1}]| / |(e_i, b_{i + 1})|$ hold for $i \in [1, N - 1]$.
Each pair $\langle b_i, s_i \rangle$ represents the gap $(e_i, b_{i + 1})$ between the segments $[b_i, e_i]$ and $[b_{i+1}, e_{i+1}]$, and its score $s_i$ is the length $|[b_i, e_{i + 1}]|$ of the connected segment resulting from closing the gap relative to the cost $|(e_i, b_{i + 1})|$ for closing it, i.e, the length of the gap.
In $\Tg$, the pairs are ordered by their scores.

As long as $\Tg \neq \emptyset$, we iteratively consider the pair $\langle b_i, s_i \rangle$ with the highest score, and remove it from $\Tg$. We check whether we can close the gap $(e_i, b_{i+1})$ it represents without exceeding $t_R$ (using a successor search over $\Ts$).
If $|R| + |(e_i, b_{i+1})| \leq t_R$, then we close the gap by merging $[b_i, e_i]$ and $[b_{i + 1}, e_{i + 1}]$ into $[b_i, e_{i + 1}]$ in $\Ts$ and possibly update the scores and starting positions of the gaps $(e_{i - 1}, b_i)$ and $(e_{i + 1}, b_{i + 2})$ in $\Tg$ (if they exist, respectively) using searches over $\Tg$.
Since we consider and search for a constant number of gaps and segments per iteration, and each search over $\Ts$ and $\Tg$ takes time $\Oh{\lg (t_R / s)}$ time, the post-processing takes overall time $\Oh{(t_R / s) \lg (t_R / s)} = \Oh{r \lg r}$.

Finally, we build $R$ by iterating once over $\Ts$ in time $\Oh{r}$.

\subsubsection{Computing the RLZ Parsing}
\label{appendix:rlzsa_parsing}

To compute the RLZ parsing of $\Ad$ w.r.t.~$R$, we build the Move-$r$ index~\cite{DBLP:conf/wea/Bertram0N24} for $\overleftarrow{R}$.
However, since we only need to compute one occurrence in $R$ per RLZ phrase, we do not construct $\MPhi$ and $\SAPhi$.
Instead, we store the array $\SAs'[1~..~r']$~\cite{DBLP:conf/wea/Bertram0N24}.
Then, we can compute exactly one occurrence $\SAs'[\hat{e}'_{y(1)}] - y(1)$ per locate query.

Suppose we have computed the parsing up to phrase $i-1$ and want to compute the $i$-th phrase.
Recall from Definition~\ref{def:rlz} that we compute the longest prefix of $\Ad[p_i ~..~ n]$ that occurs in $R$ with a backward search using Move-$r$ over $\overleftarrow{R}$.
Let $j \in [1, n - i + 1]$ be the minimum length such that $A^d[p_i ~..~ p_i + j)$ has exactly one occurrence $o$ in $R$ if it exists.
Then $s_i = o$ will produce a valid RLZ parsing, hence we can abort the backward search after $j$ steps, compute $s_i$, and instead scan for $l_i$ in $R$, i.e,\ increment $j$ until $p_i + j = n \lor s_i + j = |R| \lor \Ad[p_i + j] \neq R[s_i + j]$ in order to get $j = l_i$.

\section{Engineering Rank/Select Data Structures for Move-$r$}
\label{appendix:rank_select}

In Move-$r$, we need a data structure to answer a particular combination of rank and select queries on a string $T \in \Sigma^n$ in constant time using at most $\Oh{n}$ words of space.
Given a character $c \in \Sigma$ and a position $i \in [1, n]$, we want to either
(1) compute $\select(T, c, \rank(T, c, i) + 1)$ or
(2) compute $\select(T, c, \rank(T, c, i))$.
The former can be considered a successor query, while the latter resembles a predecessor query for occurrences of a character.

While this is a classic use case for wavelet trees~\cite{DBLP:conf/soda/GrossiGV03}, we make use of the fact that the queries only ever happen in a specific combination.
In the following, we present two simple, but practically efficient data structures that improve upon~\cite{DBLP:conf/wea/Bertram0N24}.

Let $\sigma = |\Sigma|$ be the size of the alphabet.
Our first data structure is aimed at the case where $\sigma$ is small ($\sigma = \Oh{1}$), while the second one is aimed at large alphabets ($\sigma = n^{\Oh{1}}$).
The latter is designed particularly for the construction of the RLZ parsing for RLZSA (see Appendix~\ref{appendix:rlzsa_parsing}) because there, we use Move-$r$ over $\overleftarrow{R}$ featuring a large alphabet.

\subsection{Small Alphabets}

We first consider small alphabets, i.e., we assume $\sigma = \Oh{1}$.
This scenario allows us to afford precomputing the answers for all possible query characters $c$ for a subset (sampling) of positions.
More precisely, given an integer sampling parameter $s \geq 1$, we store two two-dimensional arrays $X[1~..~\floor{n / B}][1~..~\sigma]$ and $Y[1~..~\floor{n / B}][1~..~\sigma]$ with $B = \ceil{\sigma s}$ as follows:

\begin{align*}
    X[b][c] &:= \begin{cases}
        \select(T, c, \rank(T, c, b B) + 1) & $if $c$ occurs in $T[b B ~..~ n]\\
        \infty & $otherwise$
    \end{cases}\\
    Y[b][c] &:= \begin{cases}
        \select(T, c, \rank(T, c, b B)) & $if $c$ occurs in $T[1 ~..~ b B - 1]\\
        -\infty & $otherwise$
    \end{cases}
\end{align*}
This data structure can be stored in space $\Oh{\sigma n / B} = \Oh{n / s}$ and can be constructed in time $\Oh{n + \sigma n / B} = \Oh{(1 + 1/s) n}$ time by scanning over $T$ once in both directions.

Given a query of the first type (successor) with position $i$ and character $c$, we start by computing the first block $b = \ceil{i / B}$ starting after $i$ and set $p := X[b][c]$.
Now, we scan over the preceding block in reverse and look for an occurrence of $c$, i.e., for $j$ from $b \cdot B$ to $i$, we set $p := j$ if $T[j] = c$.
Finally, we return $p$.
The second type of query (predecessor) can be answered similarly using $Y$.
This takes overall time $\Oh{B} = \Oh{\sigma s} = \Oh{1}$ and requires $\Oh{n}$ words of space since $\sigma$ and $s$ are constants.
In practice, we set $s := 4$.

We note that by using this data structure to implement $\RSLap$ in Move-$r$, we get optimal time $\Oh{m}$ for count and $\Oh{m + \occ{}}$ for locate queries in overall $\Oh{r}$ words of space.

\subsection{Large Alphabets}
\label{appendix:rank_select_large}

We now consider large alphabets, i.e., $\sigma=n^{\Oh{1}}$.
The previously shown data structure is not suitable because answering queries takes time $\Oh{\sigma s}$.
However, since the alphabet is large, we can exploit the fact that most characters occur infrequently to speed up queries.
Particularly, if a character has only few occurrences, then storing their positions in ascending order and scanning over (or performing binary search on) them for querying is sufficiently fast in practice.
For frequent characters, using an s-array~\cite{DBLP:conf/alenex/OkanoharaS07,sdsl} to marking the occurrences of in $T$ is faster.
Our rank/select data structure combines both of these methods to achieve good performance in every case.
It consists of the arrays

\begin{itemize}
    \item $C[1~..~\sigma + 1]$ with $C[c] := |\{T[i] < c \, | \, i \in [1, n]\}|$,
	\item $O[1~..~n]$ with $O[C[c] + j] := i$ such that $T[i]$ is the $j$-th occurrence of $c \in \Sigma$ and
	\item $S[1~..~\sigma]$ where $S[c]$ is an s-array marking the occurrences of $c$ in $T$, if $\#T(c) > 512$.
\end{itemize}

requiring $\Oh{n (\lg n + \mathcal{H}_0(T))}$ bits of space, where $\mathcal{H}_0(T)$ denotes the zeroth-order entropy of $T$.
The space is dominated by the s-arrays that are also used in \cite[Appendix~B]{DBLP:conf/wea/Bertram0N24} and is asymptotically no larger than a Huffman-shaped wavelet tree unless $\mathcal{H}_0(T) = \oh{\lg n}$.

We can answer $\select(T, c, i) = O[C[c] + i]$ in constant time.
To answer $\rank(T, c, i)$, we at first consider the number $o = C[c + 1] - C[c]$ of occurrences of $c$ in $T$.
If $o > 512$, we use the s-array to compute $\rank(T, c, i) = S[c].\rank_1(i)$ in time $\Oh{\lg (n / \#T(c))}$~\cite{sdsl}.
Otherwise, we compute $x := \min \{j \in [C[c], C[c+1]) \, | \, O[j] > i \}$ if it exists (using binary search if $o > 16$ or by linearly scanning otherwise).
If $x$ does not exist, then we output $o$.
Otherwise, we output $x - C[c]$.
This takes constant time because $o \leq 512 = \Oh{1}$.

\end{document}